\RequirePackage{fix-cm}
\documentclass[smallextended]{svjour3}       
\smartqed  
\usepackage{graphicx}
\usepackage[font=small,labelfont=bf,
   justification=justified,
   format=plain]{caption}
\usepackage[utf8]{inputenc}

\usepackage{amsmath,amssymb,amsfonts,graphicx}

\usepackage{epsf}

\newcommand{\eps}{\epsilon}

\newcommand{\dist}{\mathsf{dist}}
\newcommand{\Dist}{\mathsf{Dist}}

\begin{document}

\title{Testing Membership  for Timed Automata}



\author{Richard Lassaigne       \and
        Michel de Rougemont 
}


\institute{R. Lassaigne \at
              University of Paris Cit\'e, CNRS, IMJ-PRG\\
              \email{lassaigne@math.univ-paris-diderot.fr}           
           \and
           M. de Rougemont \at
              University Paris II, CNRS, IRIF\\
              \email{mdr@irif.fr} 
}

\date{Received: date / Accepted: date}

\maketitle

\begin{abstract}

 Given a timed automaton which admits thick components and a  timed word $w$, we present a tester which decides if $w$ is in the language of the automaton or if $w$ is $\eps$-far from the language, using finitely many samples taken from the weighted time distribution $\mu$ associated with the input $w$.
We introduce a distance between timed words, the {\em timed edit distance},  which generalizes the classical edit distance. A timed word $w$ is $\eps$-far from a timed  language if its relative distance to the language is greater than $\eps$.\\\\
\end{abstract}

\section{Introduction}

We study  the Membership problem of {\em timed words} for timed automata: given a non deterministic timed automaton and a timed word $w$, decide if $w$ is accepted. 
  We introduce a {\em timed edit distance} between timed words and study how to distinguish an accepted timed word from a timed word which is $\eps$-far
from the language of accepted words.  We follow the property testing approach with this new distance between timed words. Samples are taken following the weighted time distribution $\mu$ where the probability to choose a letter is proportional to its relative time delay or duration which we call its {\em weight}. The tester takes samples  which are factors of weight at least $k$, taken from the distribution $\mu$. 
Such samples can also be taken from a stream of timed words, without storing the entire input. \\

The Membership problem has been studied in \cite{AKV98,AM04} and in \cite{A18} when the input is both the automaton and the word $w$. In this case, the problem is shown to be  NP-complete.  In our situation, the automaton is fixed and the timed word $w$ is of arbitrary weight.\\

We consider timed automata  \cite{AD94} and construct the associated region automata with $m$ states.   Let $G$ be the graph whose nodes are the {\em stronly connected components} and the transient nodes of the region automaton.  In the rest of the article, a {\em component} is  a stronly connected components of $G$.  We assume the components are  {\em thick}  or of non-vanishing entropy \cite{A15}. Let $B$ be the maximum constant appearing in a  time constraint of an automaton,  let $l$ be the  number of  components and transient nodes with an outgoing transition having an unbounded guard, of a maximal path  in $G$ and $\kappa$ be a function of the number of clocks.   We require $l$ independent samples $(u_1,...u_l)$ of $\mu$, each $u_i$ is a factor of $w$ of weight $k\geq 24l.\kappa.m.B/\eps$. It guarantees that if a timed word of  total weight $T$ is $\eps$-far from the language of the timed automaton, it will be rejected with constant probability, i.e. independent of $T$.\\

Given a path $\Pi$ in $G$ and samples $(u_1,...u_l)$ from $\mu$,  definition \ref{compdef} specifies when these samples are compatible with  $\Pi$. The tester checks if there is a maximal $\Pi$ such that the $l$ samples are compatible for this $\Pi$. It rejects if no $\Pi$ is compatible with the samples. \\

The main result, theorem \ref{theo}, shows that Membership of timed words for automata with thick components is testable. First,  lemma \ref{correct} guarantees that a  word  of the language of the timed automaton is accepted by the tester. To prove that an  $\eps$-far timed word is rejected with constant probability, we  construct a corrector for a single component $C$ (lemmas 3 to 7) of the region automaton and  prove the result, i.e. theorem \ref{theoc}, in this case. We extend it to  a  sequence $\Pi$ (lemmas 8 to 13)  to prove the main result.\\

In the second section, we fix our notations  of timed automata and recall the definitions of thick and thin components. In the third section, we define the timed edit distance in the property testing context. In the fourth section, we define our membership tester. In the fifth section we give its analysis and prove the main results.


\section{Timed automata}

 Let $X$ be a finite set of variables, called clocks. A clock valuation over $X$ is a mapping $v: X \longrightarrow \mathbb{R}^+$ that assigns to each clock a time value. For each $t\in  \mathbb{R}^+$, the valuation $v+t$ is defined by $\forall x \in X ~~(v+t)(x) = v(x) +t$. A {\em clock constraint} over $X$, also called a {\em guard}, is a conjunction $g$ of atomic constraints of the form: $x \bowtie c$ where $x \in X$, $c \in \mathbb{N}$ and $\bowtie \in \{<, \leq , =, \geq ,>\}$.  Let $\mathcal{C}(X)$ be the set of all clock constraints. We write $v \models g$ when the clock valuation $v$ satisfies the clock constraint $g$ and we note $[g]$ the set of clock valuations satisfying $g$. For a subset $Y$ of $X$, we denote by $[Y \leftarrow 0]v$ the {\em reset} valuation such that for each $x \in Y$, $([Y \leftarrow 0]v)(x)=0$ and for each $x \in X\backslash Y$, $([Y \leftarrow 0]v)(x)=v(x)$. \\

 A timed automaton is a tuple $\mathcal{A} = (\Sigma ,Q , X, E, I, F)$ where $\Sigma$ is a finite set of events, $Q$ is a finite set of locations,  $I \subseteq Q$ is the set of initial locations, $F \subseteq Q$ is a set of final locations and $E \subseteq Q \times (\mathcal{C}(X) \times \Sigma \times 2^X) \times Q $ is a finite set of transitions.
 A transition is a triple 
 $(q, e, q')$ where $e= (g , a, Y) $, i.e. $g$ is the clock constraint, $a\in \Sigma$ and $Y$ is the set of clocks which are reset in the transition. Let $\kappa=2^{2^{\mid X \mid}}$ an  important parameter used in section \ref{analysis}.\\

 A state is a tuple $(q,v)$, a location $q$ and a valuation of the clocks $v$. Let $B$ be the maximum value of the constant $c$ in the atomic constraints. In this paper, $s_0$ is always the initial state, all the states are accepting and the transitions are given by figures such as Figure \ref{at0}.

\begin{figure}[h]  
\begin{center}  

\includegraphics[height=7.69cm]{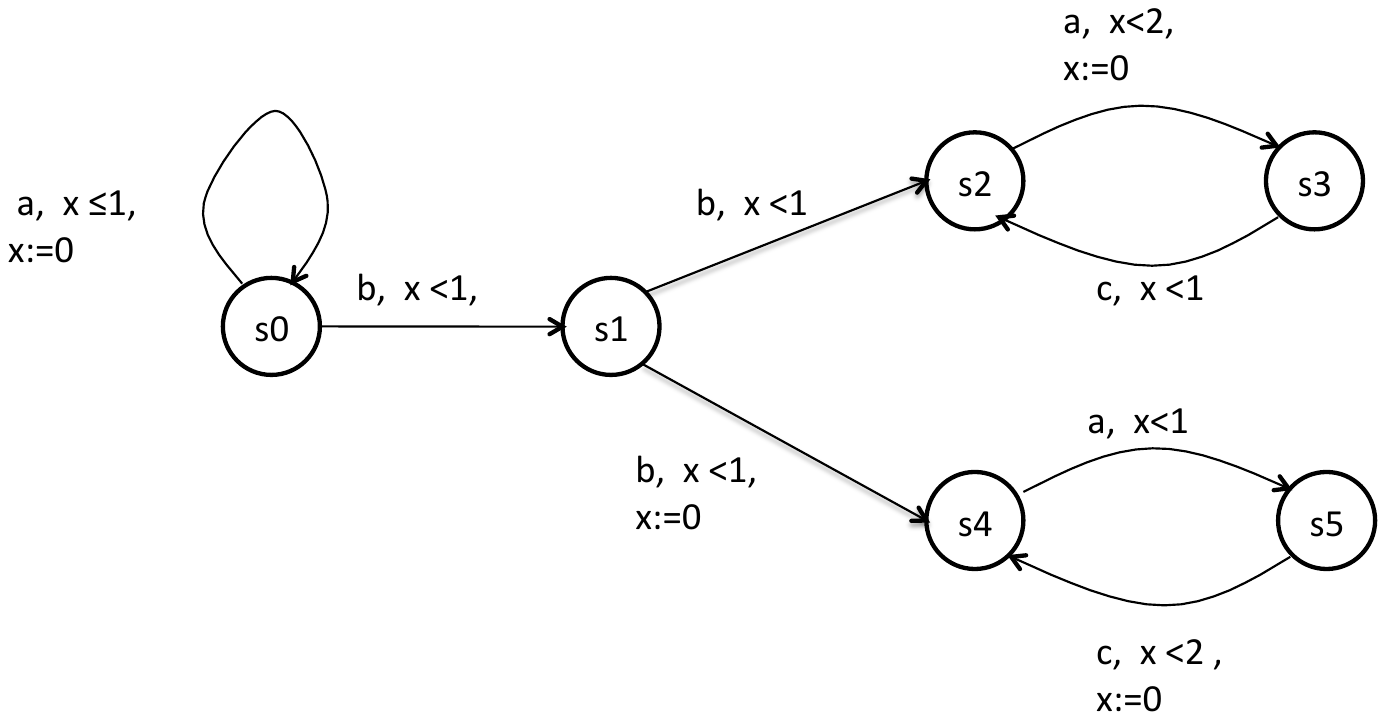}  

\caption{Timed  automaton $A_0$} \label{at0}
\end{center}  
\end{figure}

\subsection{Timed words}\label{tw}
A timed word $w$ is a sequence $(a_i,t_i)_{ 1 \leq i \leq n}$ where $a_i \in \Sigma$ and $t_i$ is a strictly monotonic sequence of values in $\mathcal{R}^+$. 
A path $\pi$ in $\mathcal{A}$ is a finite sequence of consecutive transitions:
 $(q_{i-1}, e_i, q_i)_{ 1 \leq i \leq n}$ where $(q_{i-1}, e_i, q_i) \in E$ and $e_i = (g_i , a_i, Y_i) $  where $g_i \subseteq \mathcal{C}(X)$, $a_i \in \Sigma$ and
$Y_i \subseteq X$, for each $ i \geq 0$.  The path $\pi$ is  {\em accepting} if it starts in an initial location $q_0 \in I$ and ends in a final location $q_n \in F$.
For a timed word $w$, let $untime(w)$ be the sequence of letters $(a_i)_{ 1 \leq i \leq n}$.
A {\em run of the automaton along the path $\pi$ } is a sequence:
 $$(q_0,v_0) \xrightarrow[t_1]{g_1,a_1,Y_1} (q_1,v_1) \xrightarrow[t_2]{g_2,a_2,Y_2} (q_2,v_2) \dots \xrightarrow[t_n]{g_n,a_n,Y_n} (q_n,v_n)$$
where $(a_{i}, t_i)_{ 1 \leq i \leq n}$ is a timed word and $(v_i)_{ 1 \leq i \leq n}$ a sequence of clock valuations such that:
$$(*) ~~~~\forall x \in X  ~~ v_0(x)=0 ~~ v_{i-1}+(t_i-t_{i-1}) \models g_i $$
$$ v_i= [Y_i \leftarrow 0](v_{i-1}+(t_i-t_{i-1}))   $$

We read $a_i$ for a period of time $t$ such that 
$v_{i-1} +t \models g_i$ and $v_i(x)=0$ if $x \in Y_i$, $v_i(x)=v_{i-1}(x) +t$ if $x\notin Y_i$. 
The $Y_i$ define the resets on each transition. A {\em local run} is a run where $(q_0,v_0) $
can be arbitrary.
The label of the run is the timed word $w=(a_i,t_i)_{ 1 \leq i \leq n}$, also written
$w=(a_i,\tau_i)_{ 1 \leq i \leq n}$ to use the relative time delays $\tau_i=t_i - t_{i-1}$ for $i>1$ and $\tau_1=t_1$. Such a run will be denoted by $(q_0,v_0) \xrightarrow{w} (q_n,v_n)$.\\

A timed word $w$ is  accepted by the timed automaton if it labels an accepting path $\pi$. The set of all finite timed words accepted by $\mathcal{A}$ is denoted by $L_f(\mathcal{A})$. \\

Given a timed word $w$, a factor is  a subsequence $(a_j, \tau_j)_{ j=i,i+1,....i+l}$ starting
in position $i$. Its weight is $\sum_{ j=i,....i+l} \tau_j $ and its relative weight is:

\[\frac{\sum_{ j=i,....i+l} \tau_j } {\sum_{ j=1,....n} \tau_j }\]




\subsection{Region automata}\label{ra}

Let $X$ be a set of clocks and $\mathcal{C}$  be a finite subset of $\mathcal{C}(X)$. A finite partitioning $\mathcal{R}$ of the set of valuations is a {\em set of regions} for the constraints 
$\mathcal{C}$ if the following compatibility conditions are satisfied:
\begin{enumerate}
\item $\mathcal{R}$ is {\em compatible with the constraints} $\mathcal{C}$:   for every constraint $g \in \mathcal{C}$, and every $R \in \mathcal{R}$, either $ R \subseteq [g] $ or $ [g] \cap R = \emptyset$,
\item $\mathcal{R}$ is {\em compatible with elapsing of time}:  for all $R, R' \in \mathcal{R}$, if there exists some $v \in R$ and $t \in \mathbb{R}^+$ such that $v+t \in R'$, then for every $v' \in R$, there exists some $t' \in  \mathbb{R}^+$ such that $v'+t' \in R'$,
\item $\mathcal{R}$ is {\em compatible with resets}: for all $R, R' \in \mathcal{R}$, for every subset $Y \subseteq X$, if $[Y \leftarrow 0] R \cap R' \neq \emptyset$, then $[Y \leftarrow 0] R \subseteq R' $.
 \end{enumerate}
 
 $\mathcal{R}$ defines an equivalence relation $\equiv_{\mathcal{R}}$ over valuations:
 $v \equiv_{\mathcal{R}} v' $ if for each region $R$ of $\mathcal{R}$, $v \in R \Longleftrightarrow v' \in R$. From a set of regions $\mathcal{R}$ one can define the time-successor relation: a region $R'$  is a time-successor of a region $R$ if for each valuation $v\in  R$, there exists a $t \in \mathbb{R}^+$ such that $v+t \in R'$.\\
   

Let $\mathcal{A}$ be a timed automaton with a set of constraints $\mathcal{C}$ and $\mathcal{R}$ be a finite set of regions for $\mathcal{C}$. The {\em region automaton} $\mathcal{A}_{\mathcal{R}}$ is the finite automaton defined by: 

\begin{itemize}
\item the set of states is $Q \times \mathcal{R}$,
\item the initial states are $I \times \{R_0\}$, where $R_0$ is the region containing the valuation assigning $0$ to each clock,
\item the final states are $F \times \mathcal{R}$,
\item there is a transition $(q,R) \xrightarrow {g,a,Y} (q',R')$ whenever there exists a transition $q \xrightarrow{g,a,Y} q'$ in $\mathcal{A}$ and a region $R''$ which  is a a time-successor of $R$, satisfies $g$ and $R' =  [Y \leftarrow 0] R''$.
\end{itemize}

Alur and Dill have shown \cite{AD94} how to construct a set of regions, and the size of the region automaton is  exponential in the number of clocks.
As $\mathcal{A}_{\mathcal{R}}$ is a finite automaton,  for every timed automaton $\mathcal{A}$ for which we can construct a set of regions, we can decide reachability properties using the region automaton construction. For a run of the automaton $\mathcal{A}$ of the form:
 $$(q_0,v_0) \xrightarrow{a_1,t_1} (q_1,v_1) \xrightarrow{a_2,t_2} (q_2,v_2) \dots \xrightarrow{a_n,t_n} (q_n,v_n)$$
 let its projection  be the sequence
$$(q_0,R_0) \xrightarrow{a_1} (q_1,R_1) \xrightarrow{a_2} (q_2,R_2) \dots \xrightarrow{a_n} (q_n,R_n)$$
where $(R_i)_{ 1 \leq i \leq n}$ is the sequence of regions such that $v_i \in R_i$ for each $1 \leq i \leq n$. From the definition of the transition relation for $\mathcal{A}_{\mathcal{R}}$, it follows that the projection is a run of $\mathcal{A}_{\mathcal{R}}$ over $ (a_i)_{ 1 \leq i \leq n}$.\\

Given some component $C$ of the region automaton, a timed word $w$ is {\em $C$-compatible} if there exists a local run $(q,v) \xrightarrow{w} (q',v')$ such that its projection is in $C$. Let $L_f(C)$ be the set of timed words $w$ which are $C$-compatible. \\

\begin{figure}[h]  
\begin{center}  

\includegraphics[height=7.69cm]{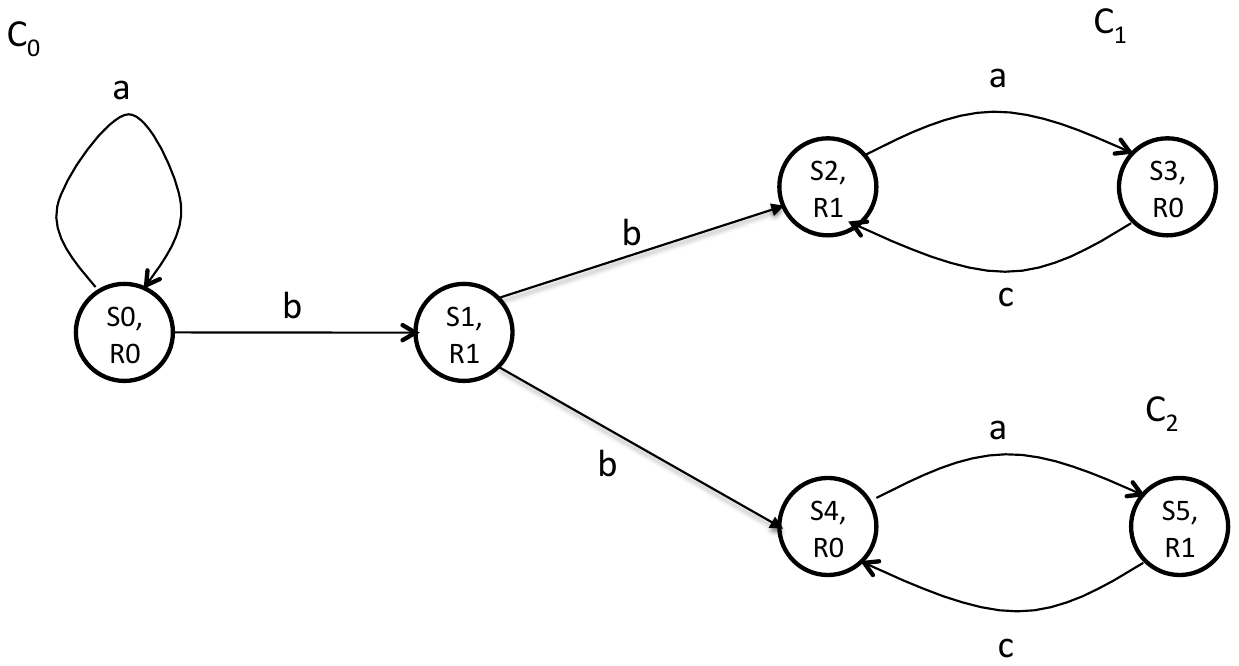} 
\centering
\caption{Region automaton of $A_0$.} 
\label{ra0}
\end{center}  
\end{figure}

The region automaton of the timed automaton $A_0$ of Figure \ref{at0} is in  Figure \ref{ra0}: the regions are $R_0: x=0$ and $R_1: 0<x<1$.  There are three components $C_0, C_1, C_2$.


\subsection{Robustness for timed systems}

Timed automata assume perfect clocks and perfect precision. The relaxation to a robust acceptance was introduced in \cite{HR00} where the reachability is undecidable. Imperfect clocks with a drift are considered in
\cite{Pur00,ATR05}, and uncertainty in the guards is introduced in \cite{DDMR08}.\\

A survey of the robustness in timed automata is presented in \cite{BMS13}.  We introduce a different approach, with a natural distance between timed words which extends to a distance between a timed word and a language $L$ of timed words.  We then show that  the approximate membership problem  becomes easy, in this setting. Precisely, we provide an $O(1)$ algorithm using an approximate decision.

\subsection{Thin and thick components}

For a component $C$, a  {\em progress cycle} is a cycle where every clock is reset on some edge of the cycle. A  {\em forgetful cycle}  $\pi_f$ is a subcase where for  all  
state $(q,R)$ on the cycle, for all $v,v' \in R$ we can find a word $w$ such that:
$$(q,v) \xrightarrow{w} (q,v') $$
By extension, given two states $(q, R)$ and $(q', R')$ of the region automaton, we can use the forgetful cycle to link two states $(q, v), v\in R$ and $(q', v'), v'\in R'$. Connect first $(q, R)$ to the forgetful cycle $ \pi_f$, follow $\pi_f$ and then connect to $(q', R')$. We can then connect $(q, v), v\in R$ and $(q', v'), v'\in R'$.\\

In \cite{A15},  components with a forgetful cycle are called {\em thick}. Components which are not thick are {\em thin }.  We assume that all the components are thick, i.e. admit forgetful cycles, and we use this hypothesis in a fundamental way. The automaton $A_0$ of Figure \ref{at0} has thick components. In contrast the automata $A_1$ and $A_2$ of Figure \ref{thin1} have distinct thin components. Their  thick components are identical. The thin component of $A_2$ has $2$ clocks $x,y$, and is called the  {\em twin thin} component in \cite{A15}.

\begin{figure}[h]  
\begin{center}  


\includegraphics[height=7.69cm]{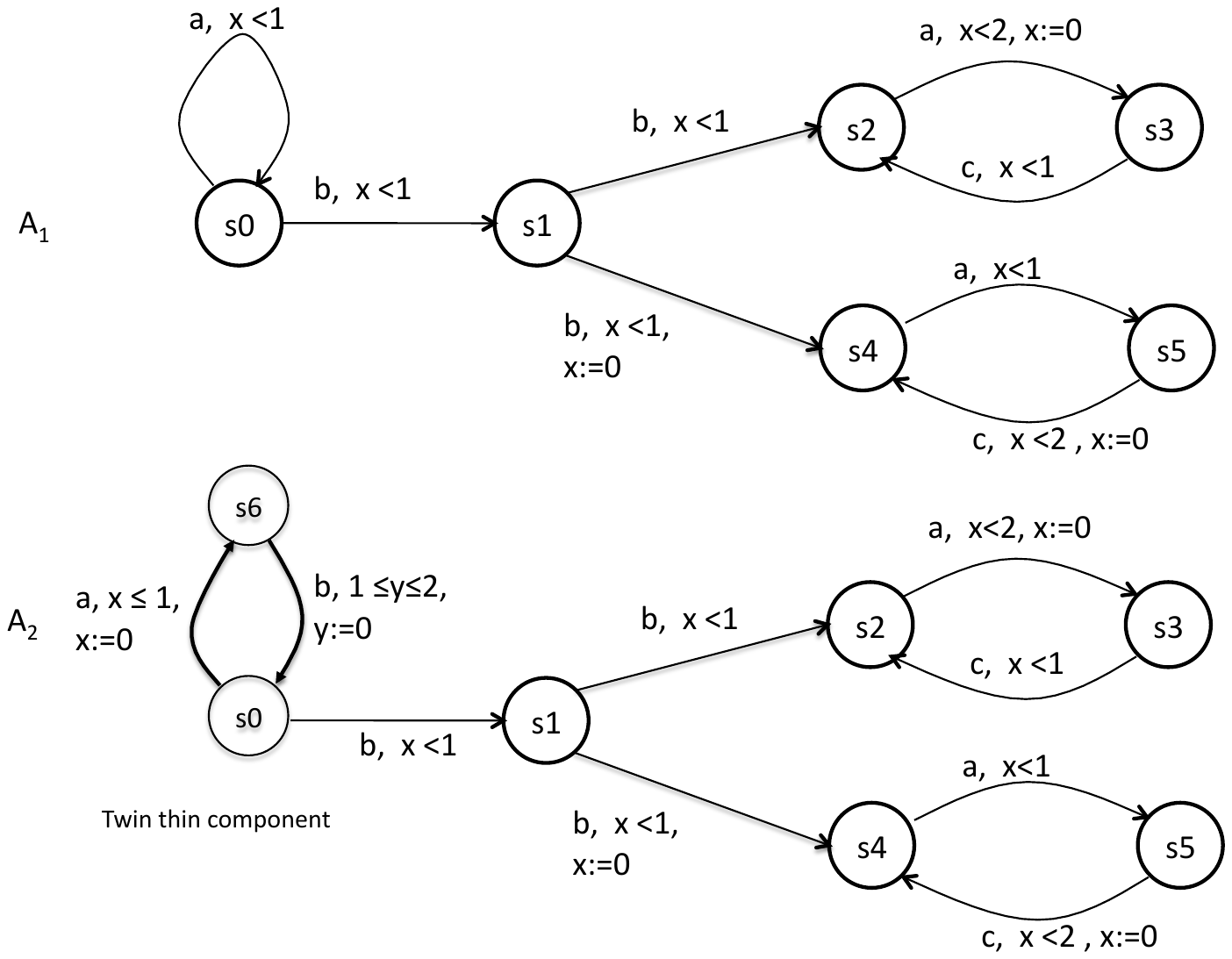}  
\caption{Timed  automata $A_1$ and $A_2$ with a thin component} \label{thin1}
\end{center}  
\end{figure}

\section{Property Testing}

For approximate decision problems,  the approximation is applied to the input
and suppose  a distance between input structures. An 
$\eps$-tester for a property $P$ accepts all
inputs which satisfy the property and rejects with high probability all inputs which
are  $\eps$-far from inputs that satisfy the property.
The approximation on the input
 was implicit in {\em Program Checking}~\cite{bk95,blr93,rs96},
in {\em Probabilistically Checkable Proofs} (PCP)~\cite{as98}, and explicitly studied for graph properties under the context of
property testing~\cite{ggr98}.\\

These restrictions allow for sublinear algorithms and even $O(1)$
time algorithms, whose complexity only depends on $\eps$.
Let $\mathbf{K}$ be a class of finite structures with a normalized
distance $\dist$ between structures, i.e.  $\dist$ lies
in $[0,1]$.
For any $\eps>0$, we say that $U,U'\in\mathbf{K}$ are {\em $\eps$-close} if their
distance is at most $\eps$. They are {\em $\eps$-far} if they are not $\eps$-close.
In the classical setting, the satisfiability of a property $P$ is the decision problem
whether $U$ satisfies $P$ for a structure $U \in \mathbf{K}$  and a
property $P\subseteq \mathbf{K}$.
A structure $U\in\mathbf{K}$ {\em $\eps$-satisfies} $P$, or  $U$ is $\eps$-close to $\mathbf{K}$
if $U$ is $\eps$-close to some $U'\in\mathbf{K}$
such that $U'$ satisfies $ P$. We say that $U$ is $\eps$-far from $\mathbf{K}$
 if $U$ is not $\eps$-close to $\mathbf{K}$.\\
 
\begin{definition}[Property Tester~\cite{ggr98}]\label{def1}
Let $\eps > 0$.
An {\em $\eps$-tester} for a property $P\subseteq\mathbf{K}$ is a randomized
algorithm $A(\eps)$ such that, for any structure $U\in\mathbf{K}$ as input:\\
(1) If $U$ satisfies $ P$, then $A(\eps)$ accepts;\\ 
(2) If $U$ is $\eps$-far from $ P$, then $A(\eps)$ rejects with probability at least $2/3$.%
\footnote{The constant $2/3$ can be replaced by any other constant $0<\gamma<1$ by
iterating $O(\log (1/\gamma))$ the $\eps$-tester and accepting iff all the
executions accept.}
\end{definition}

A {\em query} to an input structure $U$ depends on the model for accessing the
structure.
For a word $w$, a query asks for the value of $w[i]$, for some $i$.
For a tree $T$, a query asks for the value of the label of a node $i$,
and potentially for the label of its $j$-th successors, for some $j$.
For a dense graph a query asks if there exists an edge between nodes $i$ and $j$.
The {\em query complexity} is the number of  queries made to the structure.
The {\em time complexity} is the usual definition, where we assume that the
following operations are performed in constant time: arithmetic operations, a
uniform random choice of an integer
from any finite range not larger than the input size, and a query to the input.

\begin{definition}
A property $P\subseteq\mathbf{K}$ is {\em testable},  if there exists
 a randomized algorithm $A$  such that, for every real $\eps>0$ as input,
 $A(\eps)$ is an $\eps$-tester of $P$ whose
 query and time complexities depend only on $\eps$ (and not on the input size).
\end{definition}

Property testing of regular languages was first considered in~\cite{AKNS00} for the
{\em Hamming distance},
 where the Hamming distance between two words is the
minimal number of character substitutions required
to transform one word into the other.
The (normalized) edit distance between two words (resp. trees) of size $n$
is the minimal number of insertions, deletions
and substitutions of a letter (resp. node) required to transform one word (resp. tree)
into the other, divided by $n$. \\

 The testability of regular languages on words and trees was studied in \cite{MR04}
for  the edit distance with {\em moves}, that
considers one additional operation:
moving one arbitrary substring (resp. subtree) to another position in one step.
This distance seems to be more adapted in the context of property testing, since their
tester is more efficient and simpler  than the one of~\cite{AKNS00},
and can be generalized to tree regular languages.
A statistical embedding of words
which has similarities with the Parikh mapping~\cite{par66} was developped in \cite{fmr2010}.
This embedding associates to every word a sketch of constant size (for fixed $\eps$)
which allows to decide any property given by some regular grammar or even some
context-free grammar.\\

We introduce a new distance on timed words and apply the property testing framework with this distance to the membership problem of timed automata. The size of the input
is $O(n.\log T)$ if $T=t_n$ is the absolute total time and numeric time values are written in binary. Parameters of the timed automaton, the number of states $m$ and the maximum value $B$ in the constraints, are considered as constants. We assume that $n \rightarrow \infty$ and take $n$ and $T$ as the main parameters of the input. We select a fixed number of random factors of $w$, i.e. samples of size independent of $n$ and $T$ as witnesses. \\

\begin{figure}[h]  
\begin{center}  

\includegraphics[height=8.69cm]{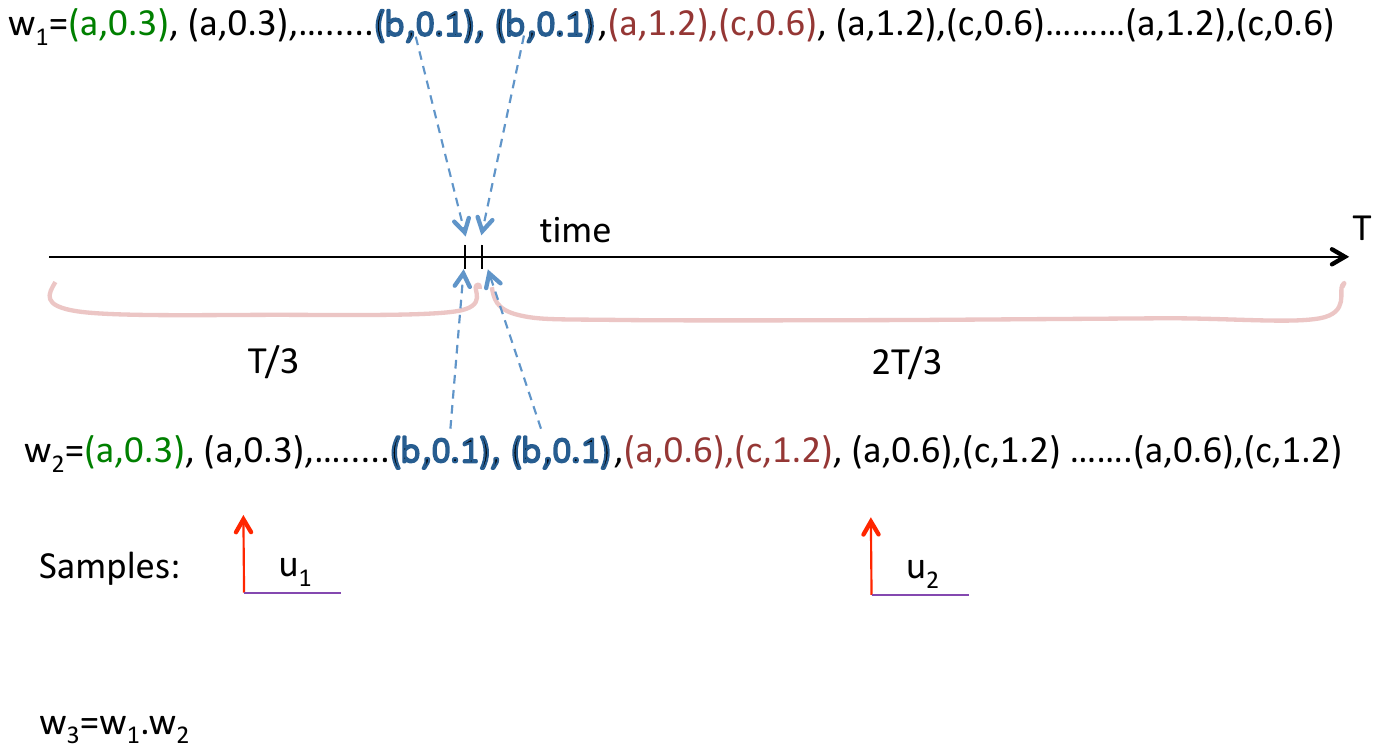} 
\centering
\caption{Close and $\eps$-far words for $A_0$: $w_1, w_2, w_3$} 
\label{runs}
\end{center}  
\end{figure}

The Figure \ref{runs} shows $3$ words:  $w_1$ of weight $T$ where we repeat first the pattern $(a,0.3)$ for a total weight of $T/3$ and later the pattern $ (a,1.2),(c,0.6)  $ for an approximate total weight of $2T/3$. 
For example:
$$w_1= (a,0.3)^{100},(b,0.1),(b,0.1),((a,1.2),(c,0.6) )^{33}$$
The weight of $w_1$ is $T\simeq 90$. 
The word $w_2$ of weight $T$  is similar except that the second iterared pattern is modified to $ (a,0.6),(c,1.2)  $. 
$$w_2= (a,0.3)^{100},(b,0.1),(b,0.1),((a,0.6),(c,1.2))^{33}$$
The word $w_3$ of weight $2T$ is the concatenation of $w_1$ followed by $w_2$. The Region automaton of $A_0$ in Figure \ref{ra0} has $3$ components $C_0,C_1,C_2$, transient states and two maximal sequences
 $\Pi_1=C_0.s_0.s_1.C_1$ and  $\Pi_2=C_0.s_0.s_1.C_2$, introduced in section \ref{tm}.  The word $w_1$ is  in $L_f(\Pi_1)$ but far from $L_f(\Pi_2)$: the pattern $ (a,1.2),(c,0.6)  $ has to be modified to $ (a,0.8),(c,0.6) $ for example. Symetrically the word $w_2$ is  in $L_f(\Pi_2)$ but far from $L_f(\Pi_1)$. Finally, $w_3$ is far from $L_f(\Pi_1)$ and from $L_f(\Pi_2)$ and therefore far from $L_f(A_0)$. We can detect these three situations with $2$ samples $(u_1,u_2)$ of finite weight with high probability, and in particular decide the potential right choice  of the non deterministic node $s_1$. In Figure \ref{runs}, $u_1$ is in the first part  $(a,0.3)^{100}$ of $w_1$ and $w_2$ with probability $1/3$ and $u_2$ in the second part with probability $2/3$.\\
 
 It is not possible to reach the same conclusion for the Automaton $A_1$ of Figure \ref{thin1}.  A finite sample on the thin component $C_0$ is not enough to witness whether or not the word is far from $L_f(A_0)$, as we may have to analyse much longer factors and dependencies between the weights of different samples. A more complex example, $A_2$ in Figure \ref{thin1} has the first component $C_0$ thin. It can be of  arbitrary high weight and therefore $2$ distinct samples may fall in this component. They become {\em dependent}, as the weights follow a long range dependency.
 
\subsection{Timed edit distance}

The classical  {\em edit distance} on words is a standard measure between two words
$w$ and $w'$. 
An  {\em  edit} operation is a {\em deletion}, an {\em insertion} or
a {\em modification of a letter}. The {\em absolute edit distance} is the minimum number
of edit operations to transform $w$ into $w'$ and the  {\em relative  edit distance} is the
absolute edit distance divided, by $\text{Max}(|w| ,|w'|)$. We mainly use the relative distance, a value
between $0$ and $1$. \\

Consider the  {\em timed edit} operations:
\begin{itemize}
\item Deletion  of  $(a,\tau)$ has cost $\tau$,
\item Insertion of $(a,\tau)$ has cost $\tau$,
\item Modification of $(a,\tau)$  into $(a, \tau')$ has cost $|\tau- \tau' |  $. 
\end{itemize}

A transformation is a sequence of operations which transform $w$ into $w'$, and the  total cost $D$ is the sum of the elementary costs.
The {\em absolute timed edit distance} $\Dist(w,w')$ between two  timed words
$w$ and $w'$ is the  minimal total cost over all possible    transformations from  $w$ into $w'$.

\begin{definition}

The {\em relative timed edit distance} between two  timed words
$w=(a_i,\tau_i)_{ 1 \leq i \leq n}$ and $w'=(a'_i,\tau'_i)_{ 1 \leq i \leq n'}$, is~:
 $$\dist(w,w')=\frac{1}{2}\cdot \frac{\Dist(w,w')}{\text{Max}(\sum_{i=1,...n} \tau_i  ,\sum_{i=1,...n'} \tau'_i )}$$
\end{definition}

If $T$ is the maximum time of $w$ and  $w'$, $\dist(w,w')$ is also $\frac{D}{2 \cdot T}$.
Two words $w,w'$ are $\eps$-close if
$\dist(w,w') \leq \eps$. The distance   between a word $w$ and   a language $L$ of timed words is defined as 
$\dist(w,L)=\text{Inf}_{w'\in L} \dist(w,w')$.  \\

{\bf Examples.} The absolute distance between  $(a,10)$  and $(a,13)$ is $3$.
 The absolute distance between $(a,10)$ and  $(b,10)$ is $20$. The absolute distance between $(a,1),(a,100)$  and
$(a,100),(a,1)$ is $2$, as shown below in section \ref{odist}.\\

To the best of our knowledge,  the  {\em relative timed edit distance} is a new distance, although other distances have been considered in the context of words with weights.
 
\subsection{Other distances}\label{odist}

The edit distance has been generalized to  a weighted edit distance where a fixed weight is associated to each letter and to
each pair of letters. The  cost of an insertion or deletion of a letter  is the weight of the letter and
the cost of a  modification  of $a$ by $b$ is the cost of the pair $(a,b)$.  For the  timed edit distance, the costs are not fixed   and depend on the positions of the letters.\\

In the context of timed words,  \cite{AM04} introduces a metric on timed words: for two words $w,w'$ of length $n$ such that 
$untime(w)=untime(w')$,  let $\dist(w,w')=\text{Max}\{| t_i-t'_i |,  0 \leq i \leq n\}$ where $t_i$ is the absolute time.  In \cite{K14}, this distance is generalized to
a vector whose first component
 captures the classical edit distance and the second component measures the maximum difference
 of the time intervals. It emphasizes the classical edit  distance between the words. As an example, the distance  of \cite{K14}   between the timed words $w=(a,1),(a,100)$  and
$w'=(a,100),(a,1)$ is the vector $(0,99)$ as the edit distance is $0$ and the 
maximum time difference is $99$.   In our framework, the absolute timed edit distance is $2$: we 
remove $(a,1)$ of the first timed word at the cost $1$ and add it after $(a,100)$ for the same cost.  An Hausdorff distance was introduced in \cite{A18}, to study similar Membership problems.\\

Another distance based on time intervals was introduced in \cite{D09}. A generalization of the classical edit distance to weighted automata was introduced in \cite{M02}. Other generalizations include the use of a permutation or of a {\em move}, the classical cut and paste. In this case a tester for the timed edit distance generalizes to these weaker distances.

\subsection{Algorithm for the timed edit distance}

The {\em absolute timed edit distance} between two words $w_1,w_2$ is computable in polynomial time by just generalizing the classical algorithm \cite{WF74} for the edit distance. Let  $A(i,j)$  be the array where
$w_1$ appears on the top row ($i=1$) starting with the empty character $\eps$, $w_2$ appears on the first column starting with  the empty character $\eps$ as in Figure \ref{ted}. 
For each letter $w$, let $w(i)$ be the relative time $\tau_i$. The value
$A(i,j)$ for $i,j>1$ is the absolute timed edit distance between the prefix of $w_1$ of length $j-2$ and the prefix of $w_2$ of length $i-2$.
Let $\Delta(i,j)=\mid \tau(i)-\tau(j)  \mid $ be the time difference between $w_1(i)$ and $w_2(j)$ if the letter symbols are identical, $\infty$ otherwise. It is the timed edit distance between two letters.\\

For $i,j >1$,  there is a simple recurrence relation  between $A(i,j),A(i-1,j),A(i,j-1) $ and $A(i-1,j-1)$, which reflects 3 possible transformations: deletion of $w_1(i-2)$, deletion of $w_2(j-2)$ or edition of the last letters. Hence:

$$A(i,j)=Min\{ A(i,j-1)+ w_1(i-2),  A(i-1,j)+w_2(j-2), A(i-1,j-1)+\Delta(i,j)  \}$$

In the example of Figure \ref{ted}, the absolute timed edit distance is 10, and we can trace the correct transformations  by tracing the minimum for each $A(i,j)$: in this case, we erase $(a,5)$ and reinsert it at the right place.
\begin{figure}[ht] 
  \centerline{\includegraphics[height=7cm]{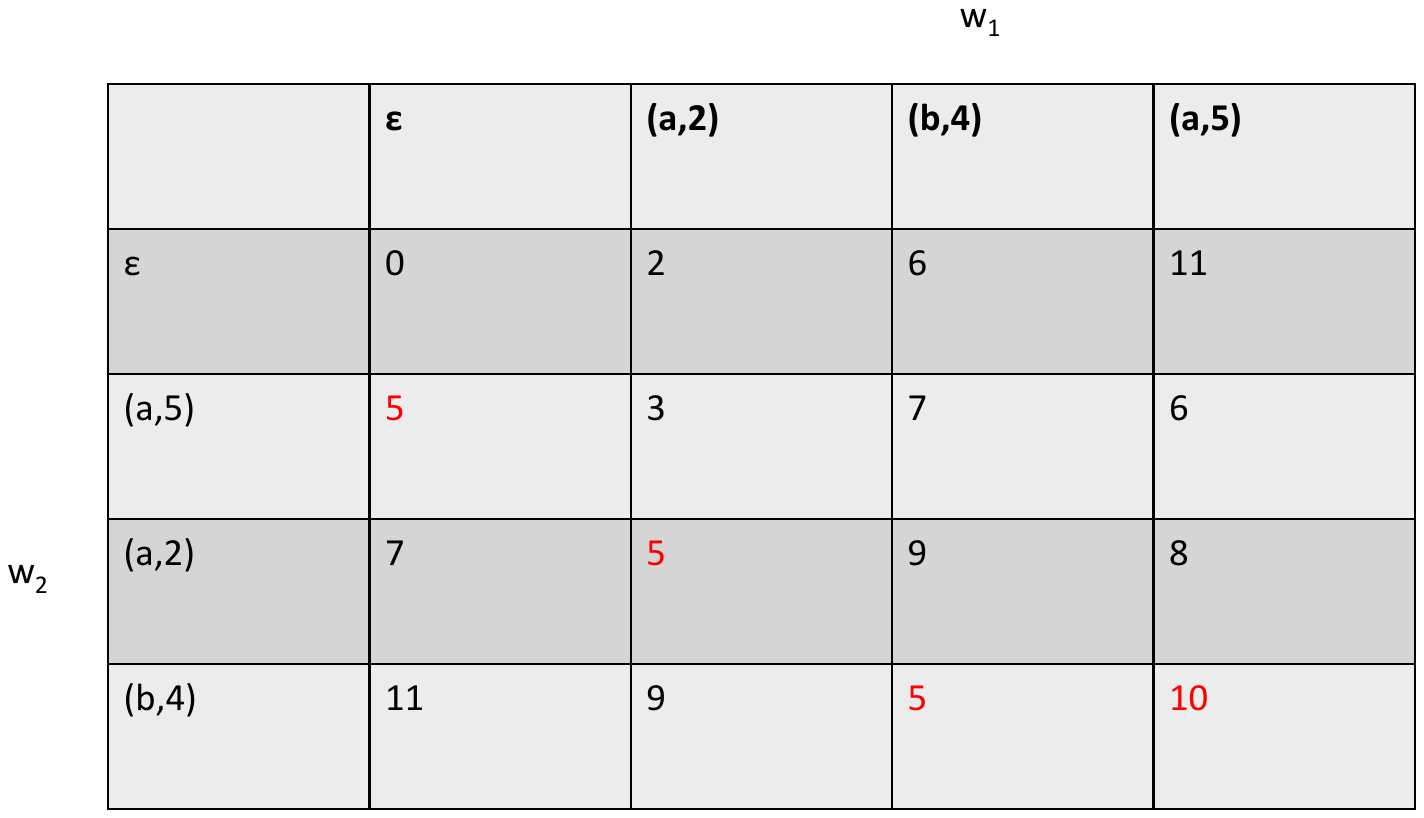}}
  \caption{Classical array $A(i,j)$ for timed edit distance between $w_1=(a,2),(b,4),(a,5)$ and $w_2=(a,5),(a,2),(b,4)$ }  \label{ted}
\end{figure}
\section{Testing  membership of timed words }\label{tm}

Given a timed word $w= (a_i, t_i)_{1 \leq  i \leq n}$, we want to {\em approximately} 
decide if $w$ is in language $L$, i.e. decide if $w$ is accepted or if  
$w$ is $\eps$-far from a language $L$, for the timed edit distance, i.e. if $\dist(w,L)\geq \eps$.  A query is specified by a weight $k$ and returns a factor of the word $w$ of weight  at least $k$ taken according to the distribution $\mu$, introduced in section \ref{mu}. This is the classical approximation taken
in Property Testing. \\


Assume the region automaton $\mathcal{A}_R$   has some components $C_i$ for $i=1,...,p$ and some  transient states $s_j$. If $v \in C$ and there is a transition from $v$ outside of $C$, we call $v$ a {\em limit node}.
Let $G$ be the graph of the components $G=(V,E)$ where the nodes $V$ are the components, the limit nodes and the  transient nodes. Edges of $E$ are:

$\bullet$ transitions  between a component $C$ and a limit node of $C$,

$\bullet$ transitions in  $\mathcal{A}_R$ between limit nodes and transient nodes,

$\bullet$ transitions in  $\mathcal{A}_R$ between two transient nodes,
 
$\bullet$ transitions  in $\mathcal{A}_R$ between a transient or limit node and another component $C'$.

\begin{definition}
A sequence $\Pi$ is a sequence of nodes of $G$ corresponding to a simple path  in $G$.
\end{definition}
We can order the sequence $\Pi$ with the natural order on the paths and speak of maximal sequences.
For the Region automaton of Figure \ref{ra0},
$\Pi_1=C_0.s_0.s_1.C_1$ and  $\Pi_2=C_0.s_0.s_1.C_2$  are the two maximal sequences.\\

An {\em extended component} $\bar{C_i}$ is a component $C_i$, possibly with a prefix of transient states. In Figure  \ref{ra0}, 
$\bar{C_1}=s_0.s_1.C_{1}$,  and $\bar{C_0}=C_{0}$. Then 
$\bar{\Pi}=\bar{C_{0}}.\bar{C_{1}}=C_{0}.s_0.s_1.C_{1}$.

\subsection{Samples and Compatibility}\label{mu}

The {\em weighted time distribution} $\mu$   selects a position   $1\leq j \leq n$ in a word $w= (a_i, t_i)_{1 \leq i \leq n}= (a_i, \tau_i)_{1 \leq i \leq n}$, i.e. a letter $(a_j, \tau_j)$, proportionally to its weight $\tau_j$:
 
$$Prob_{\mu}[j]=\tau_j/ t_n$$ 

We access the timed word with such a query which takes $1$ unit of time for the {\em query complexity} analysis.   A sample of weight at least $k$ of $w= (a_i,t_i)_{ 1 \leq  i \leq n}=(a_i, \tau_i)_{ 1 \leq  i \leq n}$ is a factor $u$ starting in position $j$ of weight   $ | u |=\sum_{i=j }^{i=j+p}\tau_i \geq k$ for the smallest possible $p$, if it exists. If we reach the end of $w$, we just have a sample of weight less than $k$. 
In practical situations, the exact time required for a query may vary. Consider two models: the model where we store the entire word $w$ and the streaming model where we read letters $(a_i, \tau_i)$ one by one and only store samples.\\

If we store the entire word, we choose a position $j$ 
by first choosing a uniform real value
 $i \in_r [0, t_n]$ and find  $j$ such that $t_{j-1}\leq i < t_j$ by dichotomy. We
first compare $i$ and $t_{n/2}$ and find the exact $j$ after at most $\log n$ steps.  There is a $O(\log n)$ overhead in this procedure.\\


In the streaming model, we can directly select $l$ distinct samples with a {\em weighted Reservoir sampling} \cite{V85} with no overhead. We take the first $l$ letters of the stream and for $j>l$ we keep the $j$-th letter with probability $l.\tau_j/t_j$. If this value is greater than $1$, we assume it is $1$. If we  keep the letter, we remove a random letter of the Reservoir with probability $1/l$ and replace it by the $j$-th letter. \\

We recall the classical argument which shows by induction on $n$, that the probability that a letter $l< j\leq n $ is in the Reservoir is
 $ l.\tau_j/ t_n$. It is true for $n=l+1$. If it is true for $n$, let us show that it is also true for $n+1$. The probability that the $j$-th letter is in the Reservoir at stage $n+1$ is:

$$ \frac{l.\tau_j}{t_n}[(1-\frac{l.\tau_{n+1}}{ t_{n+1}})+ \frac{l.\tau_{n+1}}{ t_{n+1}}.\frac{(l-1)}{l}]=\frac{l.\tau_j}{t_n}[ \frac{t_{n+1}-\tau_{n+1}}{ t_{n+1}}]= \frac{l.\tau_j}{t_{n+1}}$$

Let us justify this equality. The $j$-th letter is in the Reservoir at stage $n$ with probability $\frac{l.\tau_j}{t_n}$ by the induction hypothesis. It stays unchanged with probability $(1-\frac{l.\tau_{n+1}}{ t_{n+1}})$, when the letter $n+1$ is not kept, and with probability $\frac{l.\tau_{n+1}}{ t_{n+1}}.\frac{(l-1)}{l}$ when the letter $n+1$ is kept and the letter $j$  is not removed from the Reservoir. We extend the basic strategy of the Reservoir sampling which keeps $l$ independent samples to keep $l$ independent factors of weight $k$.  At each stage $n$, the $n$-th letter may be concatenated to the factors whose  weight has not yet reached $k$.\\

 Given a word $w$, we select $l$ independent samples of weight at least $k$ according to $\mu$, which we order as 
 $(u_1,....u_l)$. Notice that two samples  $u_i$ and $u_j$ of weight greater than $k$, where $i < j$,  may overlap. In this case, we merge the samples into a larger sample $u_i$. In the sequel, we assume all samples are disjoint which is possible if $T$ is large enough.
We now introduce the central notion of {\em compatibility} for an arbitrary  sequence $\Pi$ and a sequence of disjoint ordered factors 
$(u_1,....u_l)$ of a word $w$. \\



 For each factor $u_i$ we extend the definition of compatibility introduced in section \ref{ra} for a component $C$ to a sequence $\Pi$. We examine if it could start from some transient state $s_j$ or from some component $C_j$ of $\Pi$ and end on a transient state or on a component.\\

In section \ref{ra}, we defined the notion of a $C$-compatible timed word $w$. There is a run such that its projection is in $C$. Similarly, we can say that there is a run such that its projection is in 
$\Pi$. It starts in some transient state of $\Pi$ or in some state of a connected  component $C$ of $\Pi$ and ends  later in $\Pi$.
We say that the projection of the run follows $\Pi$.

 \begin{definition}\label{compdef}
 A sample  $u_i$ is {\em $\Pi$-compatible} from state $s=(q,R)$ to state $s'=(q',R')$  if
$  \exists ~ v, v' ~(q,v)\stackrel{u_i}{\rightarrow} (q',v')$
with $v\in R$ and $~v' \in R'$ where $s$ and $s'$ are transient states of $\Pi$ or states  of components of $\Pi$. The projection of the run must follow $\Pi$.

A sequence of disjoint ordered factors  $(u_1,....u_l)$ of a word $w$ {\em is 
$\Pi$-compatible }  if each $u_i$ is compatible for $\Pi$ from some state $s_i$ to some state $s'_i$ and for $i=1,...l-1$  $s'_i$ and $s_{i+1}$ are either in the same  component of  $\Pi$ or $s_{i+1}$ is posterior to $s'_i$ in $\Pi$.

 \end{definition}
 
Let $L_f(\Pi)$ be the set of timed words $w$ which are $\Pi$-compatible, the {\em language of $\Pi$}, and similarly for $L_f(\bar{\Pi})$. \\

 \begin{definition}
A sequence of disjoint ordered factors  $(u_1,....u_l)$  of a word $w$ is 
compatible with a timed automaton $\mathcal{A}$ if there exists a sequence  $\Pi$ such that $(u_1,....u_l)$  is 
$\Pi$-compatible.
\end{definition}


The tester takes $l$  disjoint samples $(u_1,....u_l)$, each $u_i$  is of weight  at least $2k$,
 where $k=24.l.\kappa.m.B/\eps$,  which we order 
according to their position in $w$. 




\subsection{Compatibility properties}

Let $w $ be a timed word accepted by a timed automaton $ \mathcal{A}$.
What can be said about the compatibility of samples  $(u_1,....u_l)$?

\begin{lemma}
If $w \in L_f(\mathcal{A})$, then for all $l$ and for all  disjoint ordered $l$-samples $(u_1,....u_l)$ there is a sequence $\Pi$ such that 
 $(u_1,....u_l)$ is $\Pi$-compatible.

\end{lemma}
\begin{proof}
If $w \in L_f(\mathcal{A})$, there is a run for $w$, i.e. a  sequence $\Pi$ defined by the run from the origin state to some final state $q$. The independent samples  are  $\Pi$-compatible. $\Box$
\end{proof}

Consider the following decision procedure, Algorithm $A_2$, to decide if  $(u_1,....u_l)$ is $\Pi$-compatible. 
Let $u_i$ be a factor and $\Pi$ a sequence:\\

{\bf Algorithm $A_1(u_i,\Pi)$}.
{\em Enumerate all pairs $(s_i,s'_i)$ where
$s_i$ and $s'_i$ are either a transient state of $\Pi$ or a state of a component of $\Pi$. 
If there exists a pair $(s_j,s'_j)$ such that $u_i$ is compatible for $\Pi$ from $s_j$ to $s'_j$, then Accept else Reject. }\\

 $A_1$ solves a system of linear constraints for each $(s_i,s'_i)$ where the variables are the valuations on the states from $s_i$ to $s'_i$. We  accept if there is a solution to the system and reject otherwise. We extend $A_1$  to $A_2$ which takes  $(u_1,....u_l)$ the ordered samples as input,  instead of a single $u_i$.\\

{\bf Compatibility Algorithm $A_2((u_1,....u_l),\Pi)$}.
{\em  If for each $u_i$ there exists some 
pair $(s_i,s'_i)$ such that  Algorithm $A_1(u_i,\Pi)$  Accepts and 
if for $i=1,...l-1$  $s'_i$ and $s_{i+1}$ are either in the same  component of  $\Pi$ or $s_{i+1}$ is posterior to $s'_i$ in $\Pi$, then Accept else Reject. }


\subsection{Tester}

Given a timed automaton  $\mathcal{A}$ let $L_f( \mathcal{A})$ be the language accepted. 
Let $\mathcal{A_R}$ be the region automaton with $m$ states, and let $B$ be the maximal value used in the time constraints. The automaton is fixed and the timed word $w$ is the input of weight $T$. The tester has a query and time complexity independent of $T$. 
We first generate all the maximal $\Pi$, for the inclusion,  which include components and transient states. We consider  transient  states with unbounded transitions as specific extended components.   Let
$\bar{\Pi}=\bar{C_{i_1}}.\bar{C_{i_2}}....\bar{C_{i_l}}$  the corresponding sequences of extended components obtained in this way.
 We first define a Tester   along a $\bar{\Pi}$ and the final Tester considers all possible maximal $\bar{\Pi}$.\\

{\bf Tester along a path $\bar{\Pi}=\bar{C_{i_1}}.\bar{C_{i_2}}....\bar{C_{i_l}} $ }

{\bf Input}: timed word $w$,  $\eps$,

{\bf Output}: Accept or Reject\\

{\em 
1. Sample $l$ independent disjoint  factors $(u_1 <  u_2....< u_{l})$ of weight  $k\geq 24l.\kappa.m.B/\eps$ of $w$  for  the weighted time distribution $\mu$.

2.  Accept   if $A_2((u_1,....u_l), \Pi)$ accepts else Reject.\\
}

We then obtain the general tester for a regular timed language $L_f( \mathcal{A})$.\\

{\bf Tester  for $L_f( \mathcal{A})$}

{\bf Input}: timed word $w$, $\eps$

{\bf Output}: Accept or Reject\\

{\em  1.  Construct all the $\bar{\Pi}$ corresponding to maximal $\Pi$, of the
region automaton  $\mathcal{A_R}$, starting in  the initial state,

2. For each $\bar{\Pi}$, apply the Word Tester along $\bar{\Pi}$,

3. If there is a  
$\bar{\Pi}$  such that Word Tester along $\bar{\Pi}$ accepts, then Accept
   else Reject.}\\

\section{Analysis of the Tester }\label{analysis}

We have to verify the two properties of a Tester, given in  definition \ref{def1}.

\begin{lemma}\label{correct}
If $w \in L_f( \mathcal{A}) $ then the Word Tester  for $L_f( \mathcal{A}) $   always accepts.
\end{lemma}
\begin{proof}
Consider a run   of the automaton $\mathcal{A}$, labeled by  $w$. There exists a $\bar{\Pi}$ with
at most $l $ extended components such that all the
factors $u$ of weigth $k$ of $w$ are compatible for   $\Pi$.  Any sequence of ordered factors $(u_1,...u_l)$ is also 
$\Pi$-compatible. Hence the Tester accepts. $\Box$ \end{proof}

The more difficult task is to show that if $w $ is $\eps$-far from $ L( \mathcal{A}) $ then the Tester  will reject with constant probability.  
Equivalently, we could  show the contrapositive, i.e. if the Tester accepts with constant probability, then $w $ is $\eps$-close to  $ L( \mathcal{A}) $. \\

We construct a {\em corrector} for $w$ in order to prove this property. 
A corrector  transforms an incorrect $w$ into
a correct $w'$ with  timed edit operations. We first describe a corrector for a given component $C$ in section \ref{c1}, for an extended component $\bar{C}$ in section \ref{c2} and for a path
 $\bar{\Pi}=\bar{C_1},...\bar{C_l}$ of extended components  in section \ref{c3}. In the last case, we decompose a word  $w$ into $l$ factors which we will correct for each $\bar{C_i}$.
 In each case, the corrector  shows that if a timed word is $\eps$-far from  the corresponding language, then samples of a certain weight $2k$ will be incompatible  with constant probability.\\

\subsection{Correction and Tester for a component $C$}\label{c1}

$C$ is a thick component, i.e. admits a forgetful cycle \cite{A15}. In this case, from a state $(q,R)$ and $v\in R$, we can reach any
 $(q',R')$ and $v'\in R'$ with a small timed word, called a {\em link} in Lemma \ref{corr}. We then introduce a decomposition of a word $w$ into compatible fragments separated by {\em cuts with a cost}. In Lemma \ref{close}, we show that if the total relative weight of the cuts is small, then the word is $\eps$-close to $L_f(C)$. In Theorem \ref{theoc} we show that if $w$ is $\eps$-far, samples of weight at least $2k$, a function of $\eps$, are incompatible for $L_f(C)$ with constant probability.
 
 \begin{lemma}\label{corr}
 For all pairs of states
$ (q,R), (q',R')$ of a  thick component $C$  and for all valuations $v\in R, v'\in R'$, there exists a timed word $\sigma$ such that  $ (q,v) \xrightarrow{\sigma} (q',v')$ and the weight of $\sigma$ is less than $3.\kappa.m.B. $
\end{lemma}
\begin{proof}

As $C$ is a  thick component, there is a forgetful cycle $\pi$ and a path from $(q,R)$ to  a state 
$(q_0,R_0)$ on the cycle $\pi$ and a path from $(q_0,R_0)$  to $(q',R')$.  There is a direct forgetful path $\sigma$ from   $(q,R)$ to $ (q',R')$, such that for all valuations $v\in R$ there is a
 $v'\in R'$ such that 
$ (q,v) \xrightarrow{\sigma} (q',v')$. 
The length   of $\sigma$  is less than $3.\kappa.m $ and its weight is less than
$3.\kappa.m.B $. These bounds follow the analysis of the monoid $\mathcal{M}$ of the orbit graphs \cite{A15} whose size is exponential in the number of clocks. The use of Simon's factorization \cite{Si90} adds another exponential factor in the number of clocks. $\Box$
 \end{proof}
 
 We decompose any word into compatible factors for the component $C$ and introduce the notion of a cut with a cost.  The sum $V$ of the costs of  the different cuts is the key parameter of the decomposition.\\

\begin{definition}\label{cutdef}

A {\em  cut for $C$} in a timed word  $w=(a_i,\tau_i)_{1\leq i \leq n}$  is a decomposition of $w$ into the longest possible compatible prefix $w_1$, a letter $(a_i,\tau_i)$ and a  suffix
$w'_1$, i.e. $w=w_1.(a_i,\tau_i).w'_1$, such that $w_1$ is compatible for  $C$ but $w_1.(a_i,\tau_i)$ is not compatible for $C$.  If $(a_i,\tau_i)$ as a single letter is compatible for $C$, the cut is {\em weak} otherwise the cut is {\em strong}.
\end{definition}

The correction strategy depends on the two types of cuts.
\begin{itemize}
\item In a weak cut, let $(a_i,\tau_i).w'_2$  be the longest compatible timed word from some state $(q,R)$ of $C$ of the timed word $(a_i,\tau_i).w'_1$. 
Lemma \ref{corr} provides a  link $\sigma$ before the letter  $ (a_{i},\tau_{i})$. The edit cost is then 
$c  \leq 3\kappa m B$, because 
the weight of $\sigma$ 
is less than $3\kappa m B$. 

\item In a strong cut, $(a_i,\tau_i)$  is not compatible. If   there exists  $\tau'_i \neq \tau_i $,  such that $(a_i,\tau'_i)$ is compatible, we use a link $\sigma$ and modify $\tau_i$: the correction cost is at most $c= \mid \tau_i-\tau'_i  \mid +3\kappa m B$. If it is not the case,  we erase  the letter  and the cost is $c=  \tau_i$.
\end{itemize}

\begin{figure}[h]  
\begin{center}  

\includegraphics[height=7.69cm]{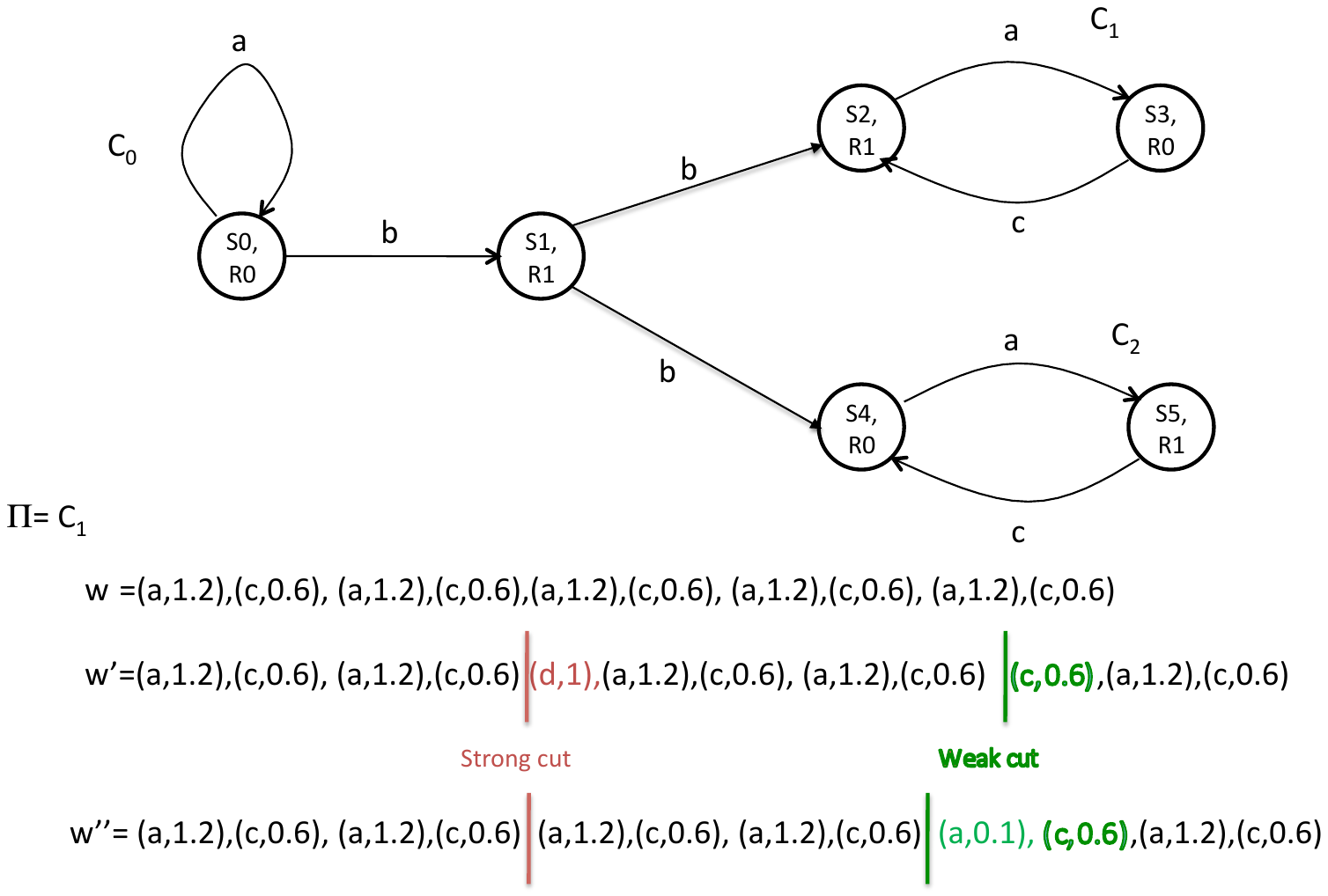}  
\caption{The word $w \in L_f( \Pi)$ for the region automaton of $A_0$ and $\Pi=C_1$, whereas $w' \notin L_f( \Pi)$: it has a strong cut and a weak cut. The word $w''$ is a possible correction of $w'$.} \label{fcuts}
\end{center}  
\end{figure}

We  then write:
$w =w_1 ~|_{c} ~w_2$ where $c$ is the cost of the correction and $w_2$ is the corrected suffix: for the strong cut, we removed a letter and for the weak cut we added a link. Figure \ref{fcuts} first  shows a strong cut and then a weak cut for the word $w'$. The corrected subwords determine $w''$.

We iterate this process on $w_2$  starting in an arbitrary state $(q,R)$ and $v\in R$.  If we decompose  $w_2 =w_{2,1} ~| ~w_{2,2}$, we write
$w =w_1~|_{c_{1}}~w_2~|_{c_{2}}~w_3$ instead of $w =w_1~|_{c_{1}}~w_{2,1}
~|_{c_{2}}~w_{2,2}$.

We now define  the algorithmic {\em decomposition}, written as:
$$w =w_1~|_{c_{1}}~w_2~|_{c_{2}}~w_3~_{c_{3}}|~...~|_{c_{h}}~w_{h+1}$$


\subsubsection{A {\em  decomposition} for a component $C$ of a timed word and its associated cost $V$}

A  $C$ {\em  decomposition} of a timed word  $w=(a_i,\tau_i)_{1\leq i \leq n}$  is the recursive process which given:
$$w =w_1~|_{c_{1}}~w_2~|_{c_{2}}~w_3~_{c_{3}}|~...w_{i}~|_{c_{i}}~w_{i+1}$$
at stage $i$ of cost $V(i)$ constructs
$$w =w_1~|_{c_{1}}~w_2~|_{c_{2}}~w_3~_{c_{3}}|~...~w_{i+1} ~|_{c_{i+1}}~w_{i+2}$$
at stage $i+1$ of cost $V(i+1)$ as follows:

If $i=1$, it is the cut construction and $V(1)=c$. If $i >1$, we assume that $w_{i+1}$ was corrected at stage $i$, with a deletion or an insertion of a link and a modification in the case of a strong cut at stage $i$, or with the insertion of a link $\sigma$ in the case of a weak cut at stage $i$. 

\begin{itemize}
\item  $w_{i+1}$  starts with an incompatible letter which needs to be  erased. At stage $i+1$ the decomposition is:

$$w =w_1~|_{c_{1}}~w_2~|_{c_{2}}~w_3~_{c_{3}}|~...w_{i} ~|_{c_{i}}~|_{c_{i+1}}~w_{i+2}$$
 and $(a_i,\tau_i)w_{i+2}=w_{i+1}$. The new $w_{i+1}$ is empty.
 The cost $V(i+1)$ is the cost $V(i)$ plus the deletion cost $c_{i+1}$.
 
 \item in all other cases, the new $w_{i+1}$ contains at least the link. It is the longest $C$ compatible segment of the old $w_{i+1}$. The new decomposition is:
 $$w =w_1~|_{c_{1}}~w_2~|_{c_{2}}~w_3~_{c_{3}}|~...w_{i} ~|_{c_{i}}~w_{i+1}~|_{c_{i+1}}~w_{i+2}$$
 The cost $V(i+1)$ is the cost $V(i)$ plus the insertion cost and the possible modification cost in the case of a strong cut, written $c_{i+1}$.\\

\end{itemize}


If  there are $h$ cuts of total value $V= \sum_{j=1...h}~ c_{j}$, for the $C$ decomposition of $w$, 
let $$c_s =\sum_{{\rm strong~  cut}~ j} ~c_{j}$$  the total cost of the strong cuts and $$c_w =\sum_{{\rm weak ~ cut}~ j}~ c_{j}$$ the total cost of the weak cuts.

\begin{lemma}\label{close}
If $w $  has  cuts of total cost $V$  for $C$, then  $w$ is  $\frac{V}{T}$-close to the  language $L_f(C)$.
\end{lemma}
\begin{proof}
Lemma \ref{corr} indicates that we can always correct a cut and start from an arbitrary new state. 
For each cut $i$ of cost $c_i$, we have a correction of  weight at most $c_i$. Hence
a total relative distance of at most  $\frac{V}{T}$ to the language of $C$.$\Box$
\end{proof}

By the contraposition of lemma \ref{close},
if $w $  of weight $T$ is  $\eps$-far from $L_f(C)$,   then $V\geq \eps.T$. We take some sample $u$ of weight at least $2k$ with the  weighted  time distribution $\mu$, for some constant $k$ we define later. We want to show that the probability that $u$ is incompatible is a constant independent of $T$, which only depends on $\eps$ and on  the automaton. We need to examine $2$ cases: either the strong cuts are dominant, i.e. $c_s \geq \eps.T/2$ or  the weak cuts are dominant, i.e. $c_w \geq \eps.T/2$. \\

\subsubsection{Dominant strong cuts.}
In a  strong cut, the cost of
the correction can be high, of the order of $\tau_{i}$.
A sample $u$ of one letter which contains
a strong cut is incompatible for $C$ whereas a sample which contains a weak cut may be compatible, as we show in section \ref{dwc}.\\

\begin{lemma}\label{sc}
If $w$ is $\eps$-far from the language $L_f(C)$ and $c_s \geq \eps.T/2$,   a sample  from the weighted time  distribution $\mu$,  has a probability  greater than $ \eps/2$ to be incompatible.

\end{lemma}
\begin{proof}
A sample $u$ of one letter taken from $\mu$ has a probability proportional to its weight. The weight of the cut is less than the weight of the letter witnessing the strong cut. If the sum of the weights of all the strong cuts is greater than $\eps.T/2$, the probability to select one is at least $ \eps/2$.$\Box$
\end{proof}

\subsubsection{Dominant weak cuts.}\label{dwc}
We now take samples as factors of weight at least $2k$, where $k$ depends on the automaton constants ($m$ and $B$) and the approximation factor $\eps$, determined in Lemma \ref{basic}. We select a starting position according to $\mu$ and the following letters until the total weight is at least $2k$. A sample which contains one weak cut may be compatible, whereas a sample which contains two consecutive weak cuts is  incompatible. 
A sample $u$ starting before the cut may be compatible for $C$ because we consider all possible states of $C$ as a starting state of $u$. \\

At the first cut we again consider all possible states  of $C$ and choose the state for which  the longest possible factor is compatible. All the runs  block before  the second cut or  precisely at the second  cut.  Hence the sample is incompatible.
We  will prove that for factors $u$ of $w$ of weight greater than $2k$, a large proportion will contain two weak cuts and hence will be  incompatible.\\

Let $\alpha_i$  for $i\geq 1$ be the number of $w_j$ in the decomposition of $w$ along the cuts, whose weight is larger
 than $2^{i-1}$ and less than $2^{i}$:
$$\alpha_i=|\{w_j:~ 2^{i-1} \leq |w_j|  <2^i \}|$$
where $|w_j|$ is the weight of $w_j$. Let $\alpha_0=|\{w_j:~ 0 \leq |w_j|  <1 \}|$.
By
 definition $h=\sum_{i \geq 0} \alpha_i$ is the total number of cuts.

A {\em small block } is a $w_j$ whose weight is smaller than  $ 24\kappa m B/\eps$, which will later be the value of $k$. Otherwise it is a {\em large block}. We need to estimate
$\sum_{0 \leq i _\leq i_l} \alpha_i $ when we choose $i_l$ as the smallest integer such that $2^{i_l}\geq 24\kappa m B/\eps$, in order to bound the number of   $w_i$ of weight smaller than $k$.
Let $$\beta= \sum_{ i  \geq i_l} \alpha_i  ~~ ~~ ~~ ~~  \gamma= \sum_{ i  < i_l} \alpha_i$$
First let us relate $\beta$, the number of large  blocks with $\gamma $, the number of small blocks when the weak cuts dominate, i.e. $c_w \geq \eps.T/2$.

\begin{lemma}(Counting cuts)\label{cut}
If $w$ is $\eps$-far from $C$ and $c_w \geq \eps.T/2$, there  is an $i_l$ such that  $\gamma \geq 3.\beta $.
\end{lemma}
\begin{proof}

There are at most $T/ 2^{i_l}$ feasible  $w_j$ of weight larger than $2^{i_l}$, i.e.
$$\beta \leq T/ 2^{i_l}    \label{beta}$$
Hence 
$h= \sum_i \alpha_i= \gamma +\beta \geq   \eps.T/6\kappa m B $ because $c_w \geq \eps.T/2$ and each weak cut has a correction of cost at most  $3\kappa m B$, as shown in Lemma \ref{corr}.

$$\gamma \geq \eps.T/6\kappa m B - \beta  $$
Let $i_l$ be the smallest integer such that 
$ 24\kappa m B/\eps \leq 2^{i_l}$. Then 
$$\beta  \leq  T/ 2^{i_l} \leq \eps.T/24\kappa m B \label{beta}~~ ~~ ~~  ~~ ~~ ~~    (*)$$
$$\gamma \geq  \eps.T/6\kappa m B - \beta  \geq \eps.T/6\kappa m B - \eps.T/24\kappa m B =\eps. T/8\kappa m B $$
$$\gamma \geq 3.\beta $$
$\Box$
\end{proof}

If we take samples of weight  $2k$, we now prove that the probability to obtain  a sample with two weak cuts, hence an incompatible sample, is greater than some constant.

\begin{lemma}\label{basic}
If $w$ is $\eps$-far from the language  $L_f(C)$  and  $c_w \geq \eps.T/2$ and $k= 24\kappa m B/\eps$, then a sample $u$ of weight $2k$ from the weighted time  distribution $\mu$ has a probability greater than $ \delta=3.\eps^2/5$ to be incompatible.
\end{lemma}
\begin{proof}
We estimate the probability that  $u$ contains two consecutive cuts, hence $u$ is incompatible.
Because $c_w \geq \eps.T/2$, as in Lemma \ref{cut}:
 
$$h=\sum_i \alpha_i \geq  \eps.T/6\kappa m B  $$

We say that {\em $u$ is in $w_j$} if the first letter of $u$ is one of the letters of $w_j$. Le us show that  if  we take a one letter sample $u$ with the weighted time distribution, then:
$$Prob[ u ~{\rm is~in}~w_j \wedge |w_{j}| \leq k]  \geq  4\eps/ 5 \label{p1}~~ ~~ ~~  ~~ ~~ ~~  ~~ ~~ ~~   (1) $$
$$Prob[ |w_{j+1}| \leq k~~ | ~~  u ~{\rm is~in}~w_j \wedge|w_{j}| \leq k] \geq 3\eps/4~~ ~~ ~~  ~~ ~~ ~~  ~~ ~~ ~~   (2) $$

For the first inequality (1), consider the correction where we erase all the small blocks, at a cost of
$Prob[ u ~{\rm is~in}~w_j \wedge |w_{j}| \leq k] .T$ and then correct at most  $\beta$ large blocks, at a cost of $\beta.3\kappa m B$. As the word is $\eps$-far, then:
  $$Prob[ u ~{\rm is~in}~w_j \wedge |w_{j}| \leq k] .T+\beta.3\kappa m B \geq \eps.T$$
  $$Prob[u ~{\rm is~in}~w_j \wedge  |w_{j}| \leq k] .T \geq \eps.T-\beta.3\kappa m B$$
We use the bound $(*)$ on $\beta$ in lemma \ref{cut}.
$$ Prob[u ~{\rm is~in}~w_j \wedge  |w_{j}| \leq k] .T \geq \eps.T-\eps.T.3\kappa m B/24\kappa m B \geq 4\eps.T/5$$
  $$Prob[ u ~{\rm is~in}~w_j \wedge |w_{j}| \leq k] \geq 4\eps/5$$

For the second  inequality (2), we have a conditional probability: we measure the probability, given that we hit a small block $j$, that the next block $j+1$ is also small. We therefore measure the probability that a sample hits a small block followed by another small block. Consider the sequences of consecutive small blocks, which exist since  $\gamma \geq 3\beta $ from lemma \ref{cut}. If we erase all the small blocks which have a small successor, we capture the small blocks followed by another small block, i.e. the event we want to measure when we take a sample of weight at least $2k$. There remains small blocks followed by large blocks, at least one small block when all the small blocks were consecutive and at most $\beta$ small blocks in the worst case.  We must therefore correct at most  $2\beta$ cuts, for the large blocks and the remaining small blocks. The cost of the erasure is: 
$$Prob[ | w_{j+1}| \leq k~~ | ~~ u ~{\rm is~in}~w_j \wedge  |w_{j}| \leq k] .T$$ and the correction cost is at most $2 \beta.3\kappa m B$. As the word is $\eps$-far, then:
  $$Prob[  | w_{j+1}| \leq k~~ | ~~  u ~{\rm is~in}~w_j \wedge |w_{j}| \leq k] .T+2\beta.3\kappa m B \geq \eps.T$$
 $$Prob[  | w_{j+1}| \leq k~~ | ~~ u ~{\rm is~in}~w_j \wedge  |w_{j}| \leq k] .T  \geq \eps.T -2\beta.3\kappa m B $$
We use the bound $(*)$ on $\beta$ in lemma \ref{cut}.
 $$Prob[  | w_{j+1}| \leq k~~ | ~~ u ~{\rm is~in}~w_j \wedge  |w_{j}| \leq k] .T\geq \eps.T- 2\eps.T.3\kappa m B/24\kappa m B$$
$$Prob[  | w_{j+1}| \leq k~~ | ~~ u ~{\rm is~in}~w_j \wedge  |w_{j}| \leq k] .T \geq 3\eps.T/4$$

We can then bound the probability that a sample of weight $2k$ is incompatible: it is greater than the probability that a sample $u$ contains $2$ successive small blocks.
$$Prob[ {\rm sample ~} u   {\rm ~of ~ weight ~} 2k {\rm ~is ~ incompatible}  ] \geq $$
$$Prob[ u ~{\rm is~in}~w_j \wedge |w_{j}| \leq k]. Prob[ |w_{j+1}| \leq k~~ | ~~ u ~{\rm is~in}~w_j \wedge  |w_{j}| \leq k] $$
$$Prob[ {\rm sample ~} u   {\rm ~of ~ weight ~} 2k {\rm ~is ~ incompatible}  ]  \geq ( 4\eps/5).(3\eps/4)\geq\delta=3\eps^2/5$$
$\Box$
\end{proof}

\begin{theorem}\label{theoc}
For a thick component $C$, the Membership problem is testable.

\end{theorem}
\begin{proof}
Lemma \ref{correct} shows the first property of the tester: if $w$ is in  $L_f(C)$, the tester always accepts. Let us show the second property:
if $w$ is  $\eps$-far from the language of  $L_f(C)$,  a random sample $u$ of weight at least $2k$ with $k= 24\kappa m B/\eps$, taken from the weighted distribution $\mu$,  is incompatible  with probability at least $3.\eps^2/5$.
If $w$ is  $\eps$-far, the total weight of the  cuts must be large by lemma \ref{close}.
Either the total weight of the strong cuts is large, greater than $\eps.T/2$, or the weight of the weak cuts is large, greater than  $\eps.T/2$. In the first case, lemma \ref{sc} shows that a one letter sample is incompatible with high probability, at least $\eps/2$.
In the second more difficult  case, lemma \ref{basic} shows that a sample of weight at least $2k$ with $k= 24\kappa m B/\eps$ contains at least $2$ cuts with high probability,  at least $3.\eps^2/5$, hence is incompatible. $\Box$
\end{proof}

We need to generalize this argument to a sequence $\bar{\Pi}$ of extended components.

\subsection{Correction and Tester for one extended component}\label{c2}

An extended component consists of transient states which may be followed by a component. Let a transition be {\em bounded } if the constraint $g$ contains an atomic constraint of the form
 $x \bowtie c$ where $x \in X$, $c \in \mathbb{N}$ and $\bowtie \in \{<, \leq , =\}$ and
{\em unbounded } otherwise.
 Transient states with unbounded transitions have to be analysed differently from transient states with bounded transitions. We therefore distinguish the two cases.\\

\begin{figure}[h]  
\begin{center}  

\includegraphics[height=8.69cm]{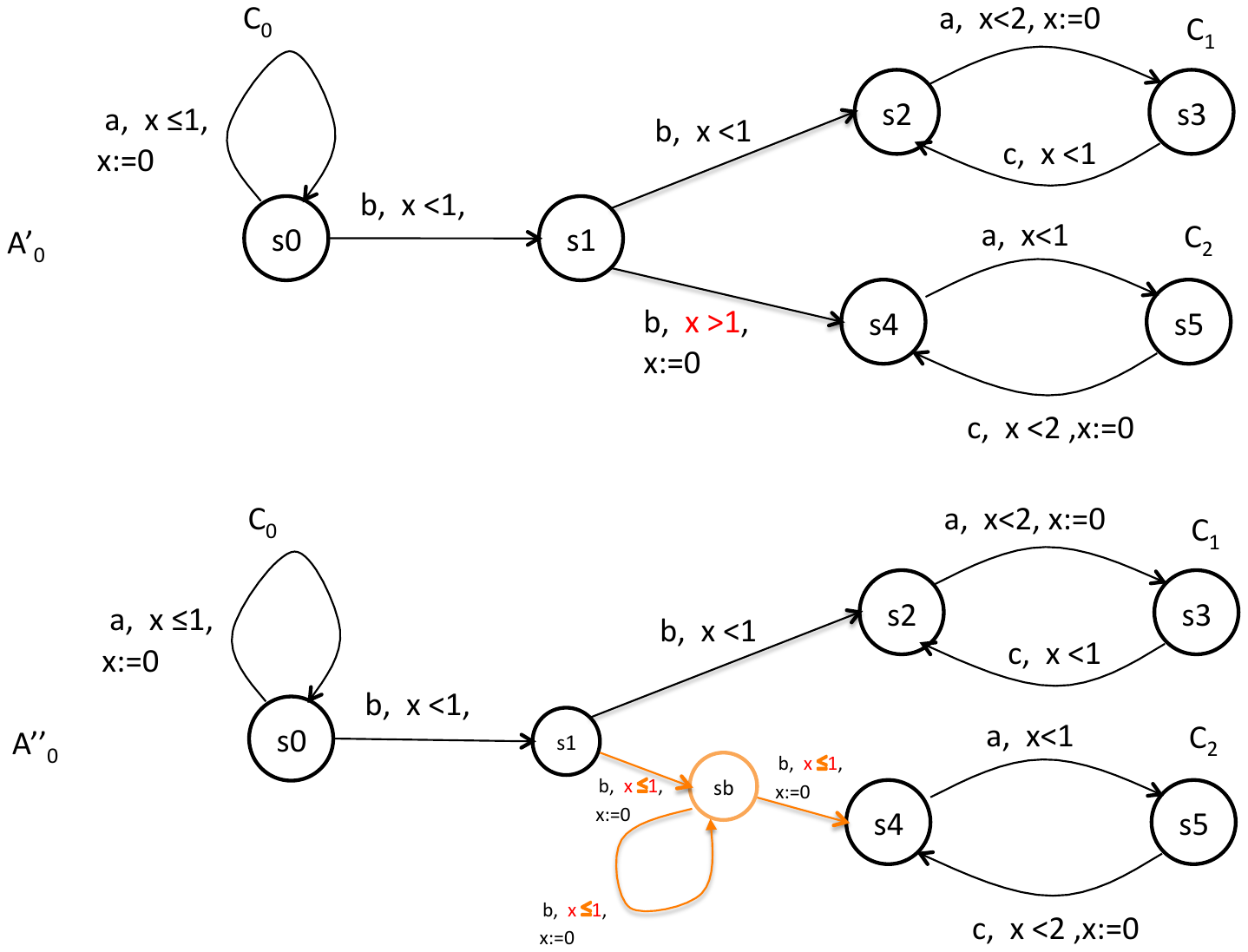}  
\caption{Automaton $A'_0$ with  an unbounded transition, and its bounded expansion $A''_0=\mathcal{T}_a(A'_0,1)$} \label{but}
\end{center}  
\end{figure}

In the example of Figure \ref{but}, the extended component $s_0.s_1.C_1$ of $A'_0$ has bounded transitions whereas its extended component $s_0.s_1.C_2$ has an unbounded transition. 

\subsubsection{Bounded transitions.}
Let $(q,R)$ be a transient state which appears in  $\Pi$. 
Let $w_1.(a_i, \tau_i).w'$ be the word $w$ where we read $(a_i, \tau_i)$  in state $(q,R)$ with the value $v\in R$.  Assume $w_1$ is compatible for a prefix  $\pi$ of $\Pi$ but 
$w_1.(a_i, \tau_i)$ is not compatible for the prefix  $\pi$  followed by the transient state $(q,R)$.
We introduce a cut before $(a_i, \tau_i)$. \\

 \begin{lemma}
The correction cost of a word $w$ for  a transient state with a bounded transition is less than $B$.
\end{lemma}
\begin{proof}
We  insert  $(a'_i, \tau'_i)$ such that $\tau'_i \leq B$  before $(a_i, \tau_i)$ which is always possible. 
The correction cost is less than $B$. $\Box$
\end{proof}

After the bounded transitions, we correct for the component $C$ as in the previous section.

\subsubsection{ Unbounded transitions}\label{unb}

 Let $(q,R)$ be a transient state with an unbounded transition which appears in  $\Pi$. 
 We show how to generalize the analysis in this context and  treat an unbounded transition whose  weight can be arbitrarly high as  a new component.
Consider the transformation $\mathcal{T}$ of a timed word $w$ relative to a bound $c$
 where letter $(b,\tau_i)$ in position $i$ in $w$ with weight $\tau_i>c$ is decomposed into $\lceil \tau_i/c \rceil$ consecutive letters $(b,c)$ with a extra letter $(b,r)$ such that 
 $\tau_i=c. \lceil \tau_i/c \rceil  +r$ where $r<c$.
 For example $(b,10.5)$ for $c=1$ is decomposed into $10$ consecutive letters $(b,1)$ followed by $(b,0.5)$, which we  write as $(b,1)^{10},(b,0.5)$. Let $\mathcal{T}(w,i,c)$ be such a transformation. \\
 
{\bf Decomposition at position $i$ for $w$, $c$ and $\mathcal{A}$.} 
 
 $\bullet$ Erase all the letters before $i$ for a global cost $c_e$.

  $\bullet$ Correct the letter $i$ for the unbounded transition at a cost $c_a$ which is at most $B+\tau_i $. In fact, if the letter is incompatible, we erase it at the cost $\tau_i$ and insert the correct letter with a weight  at most  $ B$.  If the letter is compatible, we may have to modify its weight at a cost at most $B$.
  
  $\bullet$ Correct the rest of the word $w$ for the component $C$ at  a cost $c_C$. The total cost is at most $c_e+\tau_i + B+c_C$. 
The decomposition   {\em saturates} if  $$c_e + c_a+ c_C \geq \eps.T$$\\

 Let $\mathcal{A}$ be a timed automaton where the unbounded transient transition is:
 $$b, x>c, x:=0$$
 We  define $\mathcal{T}_a(\mathcal{A},c)$  a new  timed automaton where this unbounded transition is replaced by three bounded transitions $b, x\leq c, x:=0$ with one additional state with a loop, as in Figure \ref{but}. 
Without loss of generality we consider timed automata without transition with an unbounded clock constraint without reset.

The decomposition at position $i$ for $w$, $c$ and $\mathcal{A}$ generalizes to a decomposition
at position $i$ for $\mathcal{T}(w,i,c)$, $c$ and $\mathcal{T}_a(\mathcal{A},c)$, with exactly the same costs.

 \begin{lemma}\label{tclose}
If there is a choice of the $i$-th letter where the decomposition does not saturate, then $w$ is $\eps$-close to $L_f(\bar{\Pi})$ and $\mathcal{T}(w,i,c)$ is $\eps$-close to $L_f(\mathcal{T}_a(\bar{\Pi},c)$.
\end{lemma}
\begin{proof}
We just apply the corrections associated with the decomposition. We erase all the letters
before the $i$-th letter at a cost $c_e$, adjust the $i$th letter to the unbounded transition
at a cost $c_a$ and correct the suffix for $C$ as in section \ref{c1}. The total cost is less than $\eps.T$, hence $w$ is $\eps$-close to $L_f(\bar{\Pi})$. For $\mathcal{T}(w,i,c)$ the adjustment concerns the $i$-th letter and the factor replacing $(b,\tau_i)$, with the exact same cost $c_a$. Hence $\mathcal{T}(w,i,c)$ is $\eps$-close to $L_f(\mathcal{T}(\bar{\Pi},c)$. $\Box$
 \end{proof}

We will use the contraposition: if $w$ is $\eps$-far from  $L_f(\bar{\Pi})$, then for any guess $i$, the decomposition saturates.\\

\begin{lemma}\label{transf}
If $w \in L_f(\bar{\Pi})$ then there exists a choice of a $i$-th letter of $w$  such that $\mathcal{T}(w,i,c) \in L_f(\mathcal{T}_a(\bar{\Pi}),c)$. 
If $w$ is  $\eps$-far from the language of $\bar{\Pi}$, then for any choice of a $i$-th letter of $w$,
$\mathcal{T}(w,i,c)$ is   $\eps$-far from $ L_f(\mathcal{T}_a(\bar{\Pi}),c)$.
\end{lemma}
\begin{proof}
Consider a run for $w$ in $\bar{\Pi}$. There is a letter $(b,\tau_i)$ for the unbounded transition $x>c$. We choose this letter for the transformation $\mathcal{T}$ and  write
$\tau_i=c. \lceil \tau_i/c \rceil +r$ where $r<c$. We rewrite $(b,\tau_i)$ as $(b,c)^{\lceil \tau_i/c \rceil }.(b,r)$ as a word of at least $2$ letters. We simulate the run in $ L_f(\mathcal{T}_a(\bar{\Pi}),c)$.

If $w$ is  $\eps$-far from the language of $\bar{\Pi}$, then by lemma \ref{tclose} for any choice $i$, the decomposition saturates and $\mathcal{T}(w,i,c)$ is   $\eps$-far from $ L_f(\mathcal{T}_a(\bar{\Pi}),c)$. $\Box$ 
\end{proof}

In the section \ref{c3}, we show how the
 Tester rejects with constant probability, as we consider two components and only bounded transient transitions.

\subsection{Decomposition strategy for a sequence $\bar{\Pi}$  of  $2$ extended components with bounded transient transitions} \label{c3}

We now define a correction strategy for a timed word $w$ for a sequence $\bar{\Pi}$ when
a  timed word $w$ is  $\eps$-far from the language of $\bar{\Pi}$.
 Consider the case
 $\bar{\Pi}=\bar{C_1}, \bar{C_2}$, which we can later  generalize to an arbitrary $\bar{\Pi}$. The decomposition splits the word into $2$ parts and we apply the correction strategy for $\bar{C_1}$ on the first part $I_1$ and the correction strategy for $\bar{C_2}$ on the second part $I_2$. We define the {\em border} as the position of the  cut which  partitions the word $w$ into  $(I_1,I_2)$, i.e. the last cut for the correction on $C_1$.\\
 
\subsubsection{Decomposition strategy for a sequence $\bar{\Pi}=\bar{C_1}, \bar{C_2}$}
If a timed word $w$ is $\eps/2$-far from the language of $\bar{\Pi}$, there  is an 
$(I_1,I_2)$ decomposition, defined as follows.

\begin{itemize}
\item 
Start in any state of  $C_1$ and take the longest compatible prefix $w_1$. It determines a cut of cost $c$  with the  corrector for $C_1$. We continue  in a similar way  and accumulate the costs of the corrections. If we reach a cut whose total weight $V_c$ is at least $\eps.T/2$   the {\em border } is the position of this cut and $I_1$ is 
 the prefix and $I_2$  is the suffix.
 
 \item After we reach the border, we correct  the  word for $C_2$. If the total cost of the corrections reaches $\eps.T$, we say that {\em the decomposition saturates $w$} for $\bar{\Pi}=\bar{C_1}, \bar{C_2}$ .
\end{itemize}
 
The goal is to guarantee that at least $\eps.T/2$ error occurs for $\bar{C_1}$ in $I_1$ and  for $\bar{C_2}$ in $I_2$ if the word $w$ is $\eps$-far.

\begin{lemma}\label{equi}
If a  timed word $w$ is  $\eps$-far from the language of $\bar{\Pi}=\bar{C_1}, \bar{C_2}$ then  there is a border and  the decomposition saturates $w$.
\end{lemma}
\begin{proof}
By contraposition, we consider two cases. If  there is no border, then $w$ is $\eps$-close to $\bar{C_1}$ hence to $\bar{C_1}, \bar{C_2}$. If there is a border and the decomposition does not saturate then $w$ is close to $\bar{C_1}, \bar{C_2}$, as we can find a correction of total cost less than $\eps.T$.$\Box$
\end{proof}





\subsubsection{Tester for  $\bar{\Pi}$ of $2$ extended components with bounded transient transitions} \label{t2comp}

If $w$ is $\eps$-far from $\bar{\Pi}=\bar{C_1}, \bar{C_2}$, then there are many samples
$u$'s of weight $2k$ incompatible  for $C_1$ in the interval $I_1$ and many samples
$u$'s of weight $2k$ incompatible  for $C_2$ in the interval $I_2$. We write $u\in I_1$
to indicate that the sample $u$ is a subword of $I_1$. We  conclude,  in lemma \ref{main} below, that the Tester will reject with constant probability. 

\begin{lemma}\label{main}
If $w$ is  $\eps$-far from the language of $\bar{\Pi}=\bar{C_1}, \bar{C_2}$ and $k=48\kappa m B/\eps$,  then the Tester  along $\bar{\Pi}$ rejects with constant probability.
\end{lemma}
\begin{proof}
Assume $w$ is  $\eps$-far from $\bar{\Pi}=\bar{C_1}.\bar{C_2}$. By lemma \ref{equi}  we have a decomposition $(I_1,I_2)$.
Consider two distinct samples $u_1 <u_2$ taken independently of weight at least $2k$.    If $u_1$ is incompatible  for $C_1$
and $u_2$ is incompatible  for $C_2$, then the Tester rejects. Hence:

$$Prob[\mathrm{Tester ~ rejects}] \geq Prob[u_1 \mathrm{~incompatible ~ for ~C_{1}}\wedge  u_2 \mathrm{~incompatible ~ for ~C_{2}} ] $$
$$   \geq Prob[ (u_1 \in I_1 \wedge  u_1 \mathrm{~incompatible ~ for ~C_{1}})\wedge  (u_2\in I_2 \wedge  u_2  \mathrm{~incompatible ~ for ~C_{2}} ]  $$
These two events are independent because the samples are independent, hence we can rewrite the expression as:
$$    Prob[ (u_1 \in I_1 \wedge  u_1 \mathrm{~incompatible ~ for ~C_{1}})] .Prob[ (u_2\in I_2 \wedge  u_2  \mathrm{~incompatible ~ for ~C_{2}} ] $$

Let $\eps'=\eps/2$. From the Theorem \ref{theoc}, if
$k=24\kappa m B/\eps'=48\kappa m B/\eps$, then  
$$ Prob[ u_1 \mathrm{~incompatible ~ for ~C_{1}}| u_1 \in I_1] \geq 3\eps'~^2/5=3\eps^2/20$$
and
$$    Prob[ u_1 \in I_1 ] \geq (\eps/2)$$
Hence:
$$    Prob[ (u_1 \in I_1 \wedge  u_1 \mathrm{~incompatible ~ for ~C_{1}})] \geq (\eps/2). 3\eps^2/20$$
and similarly for $u_2$.
Hence:
$$Prob[\mathrm{Tester ~ rejects}] \geq  (3\eps^3/40)^2$$
$\Box$
\end{proof}

\subsection{Decomposition strategy for a sequence $\bar{\Pi}$ of $1$ or $2$ extended components with unbounded transient transitions} \label{c3u}

We apply the transformation of section \ref{unb} and consider the unbounded transient transitions as new components. We therefore have at least $2$ components and in addition a new component for each unbounded transition. We study this general case in section \ref{c3g}.

\subsection{Correction strategy and Tester for an arbitrary sequence $\bar{\Pi}$ of length $l$} \label{c3g}
Let $\bar{\Pi}=\bar{C_{1}}....\bar{C_{l}}$ be a sequence where $l$ is the number of components plus the number of unbounded transient transitions.
The decomposition
generalizes to $\bar{C_{1}}....\bar{C_{l}}$ by taking cuts for $\bar{C_{1}}$ of global cost greater than $\eps.T/l$, until a first border and a prefix $I_1$, taking cuts for $\bar{C_{2}}$ of global cost  at least $\eps.T/l$ until a second border $I_2$ and so on until   possible  cuts for $\bar{C_{l}}$ of global cost at least $\eps.T/l$ and a last border $I_{l-1}$.

  The  decomposition $(I_1,I_2,...I_{l-1})$ {\em saturates $w$} if  the total cost  is larger than $\eps.T$.\\


\begin{lemma}\label{equi1}
If a  timed word $w$ is  $\eps$-far from the language of $\bar{\Pi}=\bar{C_{1}}....\bar{C_{l}}$ then  there are $l-1$ borders  and  the decomposition saturates $w$.
\end{lemma}
\begin{proof}
By contraposition, we consider two cases. If  there  are less than $l-1$ borders,  then $w$ is $\eps$-close to a prefix of $\bar{\Pi}$ hence to $\bar{\Pi}$. If the decomposition does not saturate then $w$ is close to $\bar{\Pi}$, as we can find a correction of total cost less than $\eps.T$. $\Box$
\end{proof}


\begin{theorem}\label{theo}
\label{WordTesterReject}
If $A$ is a timed automaton with thick components only, the Membership problem is testable.
\end{theorem}
\begin{proof}
Lemma \ref{correct} shows the first property of the tester: if $w$ is in  $L_f(A)$, there is a  $\Pi$ such that   $w$ is in  $L_f(\Pi)$ and the tester for this $\Pi$ always accepts. Let us show the second property:
If $w$ is  $\eps$-far from the language $L_f(A)$, it is also  $\eps$-far from all $L_f(\Pi)$.
Let $\bar{\Pi}=\bar{C_1}...\bar{C_l}$ and $k=24l.\kappa m B/\eps$: let us show that   the Tester along the path   $\bar{\Pi}$ rejects with constant probability.
If $w$ is  $\eps$-far from $\Pi=\bar{C_{1}}....\bar{C_{l}}$, there is a decomposition $(I_1,I_2,...I_l)$ with non-empty segments by the lemma \ref{equi1}, the decomposition saturates $w$.
Consider $l$ samples $u_1 <u_2 <....<u_l$ taken independently of weight at least $2k$ which do not overlap.  
 If $u_1$ is incompatible  for $\bar{C_1}$
and $u_2$ is incompatible  for $\bar{C_2}$.... and $u_l$ is incompatible  for $\bar{C_l}$, then the Tester rejects. Hence:
$$Prob[\mathrm{Tester ~ rejects}] \geq Prob[u_1 \mathrm{~incompatible ~ for ~\bar{C_1}}\wedge  u_2 \mathrm{~incompatible ~ for ~\bar{C_2}} .... $$
$$ \wedge u_l \mathrm{~incompatible ~ for ~\bar{C_l}} ] $$
$$  Prob[ (u_i \mathrm{~incompatible ~ for ~\bar{C_i}})   \geq Prob[ (u_i \in I_i \wedge  u_i \mathrm{~incompatible ~ for ~\bar{C_i}}) ] $$
All these  events are independent, hence we can rewrite the expression as:

$$ \prod_i Prob[ (u_i \in I_i \wedge  u_i \mathrm{~incompatible ~ for ~\bar{C_i}}) ] $$

Let $\eps'=\eps/l$. From the lemma \ref{basic}, if
$k=24\kappa m B/\eps'=24l.\kappa m B/\eps$, then  
$$ Prob[ u_i \mathrm{~incompatible ~ for ~\bar{C_i}}| u_i \in I_i] \geq 3\eps'^2/5=3\eps^2/5l^2$$
and
$$    Prob[ u_i \in I_i ] \geq (\eps/l)$$

Hence:
$$  Prob[ (u_i \in I_i \wedge  u_i \mathrm{~incompatible ~ for ~\bar{C_i}}) ] \geq  (\eps/l).3\eps^2/5l^2=3\eps^3/5l^3$$
$$Prob[\mathrm{Tester ~ rejects}] \geq \prod_i Prob[ (u_i \in I_i \wedge  u_i \mathrm{~incompatible ~ for ~\bar{C_i}}) ]\geq  (3\eps^3/5l^3)^l$$
The Tester rejects with a probability greater than a function of $\eps$ and $l$, but independent of $T$.
$\Box$
\end{proof}
\section{Conclusion}
We introduced the {\em timed edit distance} between timed words and use it  to test the membership property of timed automata with thick components. We   select factors of a timed word proportional  to their weight  according to the weighted time  distribution $\mu$.  We followed the property testing framework and constructed a tester which selects finitely many samples of bounded weight to detect if a  timed word is accepted  or $\eps$-far from the language of the timed automaton. \\

In a streaming context,  {\em the weighted  samples}  can be taken online, and the tester provides an approximate decision to the Membership problem. \\

This work can be extended in several directions.  A  first direction  is the comparison of  the languages of two timed automata, an undecidable problem \cite{AD94,AM04}.
Let us say that two timed automata $\mathcal {A}_1$ and $\mathcal {A}_2$ are $\eps$-close  if any timed word of  $L(\mathcal {A}_1)$ is $\eps$-close to $L(\mathcal {A}_2)$ and symmetrically.
A natural question is to decide if two timed automata are close or far for finite words, as it is studied in \cite{fmr2010} in the case of classical automata. A second  direction would be to generalize to infinite words.
The distance can be extended to infinite words by taking the limits of the distance on their prefixes.
We can then ask for efficient probabilistic algorithms which approximate the equivalence of timed automata for finite and infinite timed words. \\

{\bf Acknowledgment.} The authors thank E. Asarin for pointing out the concept of forgetful cycles and  thick components \cite{A15} and Aldric Degorre for the bound on the weight of a forgetful cycle. We also thank the two reviewers for their constructive comments.

\bibliography{reference1}
\bibliographystyle{plain}
\end{document}